\documentclass[a4paper,twocolumn,11pt,accepted=2021-12-27]{quantumarticle}
\pdfoutput=1
\usepackage[utf8]{inputenc}
\usepackage[english]{babel}
\usepackage[T1]{fontenc}
\usepackage{amsmath}
\usepackage{hyperref}

\usepackage{tikz}
\usepackage{lipsum}
\usepackage{graphicx}% Include figure files
\usepackage{dcolumn}% Align table columns on decimal point
\usepackage{bm}% bold math
\usepackage{lipsum} % Used for inserting dummy 'Lorem ipsum' text into the template
\usepackage{amsfonts,amssymb,bm}
\usepackage{epsfig}
\usepackage{mathrsfs}
\usepackage{color,soul}
\usepackage{bbold}
\hypersetup{colorlinks,%
citecolor=black,%
filecolor=black,%
linkcolor=black,%
urlcolor=black}
\usepackage{pstricks}
\usepackage{multirow}
\usepackage{graphicx}
\usepackage{textcomp}
\usepackage{braket}
\usepackage{empheq}
\hypersetup{colorlinks,%
citecolor=black,%
filecolor=black,%
linkcolor=black,%
urlcolor=black}
\usepackage{xcolor}
\usepackage{rotating}
\usepackage{amsthm}
\usepackage{eurosym}
\usepackage{siunitx}
\usepackage{rotating}
\usepackage{eurosym}
\usepackage[babel]{csquotes}
\usepackage{soul}

\usepackage[font=small]{caption}
\usepackage{mwe}

\newtheorem*{theorem}{Theorem}
\usepackage[framemethod=default]{mdframed}
\usepackage[labelfont=bf,labelsep=quad,justification=raggedright]{caption}
\usepackage{newfloat}

\DeclareFloatingEnvironment[name={S}]{suppfigure}
\DeclareFloatingEnvironment[name={S}]{supptable}
\allowdisplaybreaks
\usepackage{enumitem}
\usepackage[numbers,sort&compress]{natbib}

\begin{document}

\title{Experimental entanglement of temporal order}

\author{Giulia Rubino}
\affiliation{Vienna Center for Quantum Science and Technology (VCQ), Faculty of Physics, University of Vienna, Boltzmanngasse 5, Vienna A-1090, Austria}
\affiliation{Quantum Engineering Technology Labs, H. H. Wills Physics
Laboratory and Department of Electrical \& Electronic Engineering,
University of Bristol, Bristol BS8 1FD, United Kingdom}
\orcid{0000-0002-4249-522X}
\email{giulia.rubino@bristol.ac.uk}

\author{Lee A. Rozema}
\affiliation{Vienna Center for Quantum Science and Technology (VCQ), Faculty of Physics, University of Vienna, Boltzmanngasse 5, Vienna A-1090, Austria}

\author{Francesco Massa}
\affiliation{Vienna Center for Quantum Science and Technology (VCQ), Faculty of Physics, University of Vienna, Boltzmanngasse 5, Vienna A-1090, Austria}

\author{Mateus Ara\'ujo}
\affiliation{Vienna Center for Quantum Science and Technology (VCQ), Faculty of Physics, University of Vienna, Boltzmanngasse 5, Vienna A-1090, Austria}
\affiliation{Institute for Quantum Optics \& Quantum Information (IQOQI), Austrian Academy of Sciences, Boltzmanngasse 3, Vienna A-1090, Austria}

\author{Magdalena Zych}
\affiliation{Centre for Engineered Quantum Systems, School of Mathematics and Physics, The University of Queensland, St Lucia, QLD 4072,
Australia}

\author{\v{C}aslav Brukner}
\affiliation{Vienna Center for Quantum Science and Technology (VCQ), Faculty of Physics, University of Vienna, Boltzmanngasse 5, Vienna A-1090, Austria}
\affiliation{Institute for Quantum Optics \& Quantum Information (IQOQI), Austrian Academy of Sciences, Boltzmanngasse 3, Vienna A-1090, Austria}

\author{Philip Walther}
\affiliation{Vienna Center for Quantum Science and Technology (VCQ), Faculty of Physics, University of Vienna, Boltzmanngasse 5, Vienna A-1090, Austria}
\orcid{0000-0002-4249-522X}

%\date{\today}
\begin{abstract}
  The study of causal relations has recently been applied to the quantum realm, leading to the discovery that not all physical processes have a definite causal structure. While indefinite causal processes have previously been experimentally shown, these proofs relied on the quantum description of the experiments. Yet, the same experimental data could also be compatible with definite causal structures within different descriptions. Here, we present the first demonstration of indefinite temporal order outside of quantum formalism. We show that our experimental outcomes are incompatible with a class of generalised probabilistic theories satisfying the assumptions of locality and definite temporal order. To this end, we derive physical constraints (in the form of a Bell-like inequality) on experimental outcomes within such a class of theories. We then experimentally invalidate these theories by violating the inequality using entangled temporal order. This provides experimental evidence that there exist correlations in nature which are incompatible with the assumptions of locality and definite temporal order.
\end{abstract}
\maketitle

\section{Introduction}

Bell's theorem revolutionized the foundations of physics, leading to experiments which could demonstrate that nature cannot be described by a local-causal theory, and paving the way for modern quantum information~\cite{Bell1964, RevModPhys.86.419}.
One of the strengths of Bell's theorem is that it allows one to draw conclusions about nature without referring to the underlying physical theory.
%One of the main strengths of the theorem is that it is formulated without reference to any particular theory. Because of this `theory-independence', experimental loophole-free violations of Bell's theorem not only show that quantum mechanics is not locally realist, but indeed that nature herself cannot be described by any locally realist model \cite{Nature.526.682, PhysRevLett.115.250402, PhysRevLett.115.250401}.
This is a crucial feature, since ``local causality'' is conceived as a hypothesis about a fundamental property of nature, and, as such, its experimental violation reveals a statement about how nature must (or must not) be, and not just about any specific theory used to describe it.

Over the past decades, tests of Bell's theorem have been performed with
%proven for
many different physical systems thereby entangling various observables [such as spin~\cite{PhysRevD.14.2543, Nature.409.791, Nature.526.682}, polarization~\cite{PhysRevLett.28.938, PhysRevLett.49.1804, PhysRevLett.115.250402, PhysRevLett.115.250401}, position~\cite{PhysRevLett.92.210403}, and energy~\cite{PhysRevA.47.R2472, PhysRevLett.64.2495}] of two or more particles.
%However, since there is no observable associated to the measurement of the temporal order between events, this test has never been applied to the study of causal structures.
All these experimental tests involve two or more parties making measurements on entangled particles in distant laboratories, and determining their correlations. These tests prove that, in general, correlations between space-like separated events cannot be explained on the basis of a past common cause and local choices of measurements. However, since these correlations are always no-signalling (i.e., the statistics observed in one laboratory are independent of the choices of measurements made in other laboratories), Bell's theorem and related experimental tests do not address specific causal relations between events in which an experimental intervention may influence one of the events (i.e., signaling correlations). This stronger notion of causality~\cite{NatPhys.10.259} is at the heart of the modern research field of indefinite quantum causality, and is the subject of the present experimental test.

Thus far, in all established physical theories, it was assumed that the order between events is well-defined. This means that, that for any two causally related events $A$ and $B$, either $A$ signals to $B$, or $B$ signals to $A$. In the language of quantum information, these two events could be seen as the input and output of a quantum channel.
%Nevertheless, recent work has provided techniques to investigate this topic both theoretically \cite{NewJournPhys.17.102001, Oreshkov_2016NJP} and experimentally \cite{Rubinoe1602589}.
%The role of causal order in quantum mechanics has recently been studied, and 
However, it has recently been realized that quantum mechanics also allows for the existence of processes other than quantum channels, i.e., processes that are neither causally ordered, nor a probabilistic mixture of causally ordered processes. For these processes, it is genuinely indefinite whether $A$ signals to $B$ or $B$ signals to $A$, and so they are called processes with an indefinite causal structure.
%For example, in quantum mechanics, quantum channels and quantum states are processes with a definite causal order, meaning that they enable either one-way-signalling (i.e., from a `cause' to an `effect') quantum channels, or no-signaling.
Under ``processes,'' we define the set of causal relations between operations performed in different local laboratories~\cite{PhysRevA.88.022318, NatComm.3.012316, NatPhys.10.259}.
%In other words, these processes cannot be understood as one-way-signalling (i.e., from a `cause' to an `effect') quantum channels, quantum states, or any convex mixture of them \cite{PhysRevA.88.022318, NatComm.3.012316, NatPhys.10.259}. 
More precisely, a quantum process is called causally separable if it can be decomposed as a convex combination of causally ordered processes, otherwise it is causally non-separable. 
(Note that the term ``temporal'' order is used here to refer to the order among operations which cannot be used to receive signals --- in particular, to unitary ones --- whereas ``causal'' order refers to more general operations which allow both receiving and sending signals between laboratories.)
%(Note that the term `temporal' order is used to refer to events which do not include measurements --- and thus cannot be used to receive signals --- but rather unitary operations, whereas `causal' order refers to more general operations.) 

Recently, a method for certifying causal non-separability, based on ``causal witnesses,'' was developed~\cite{NewJournPhys.17.102001, Oreshkov_2016NJP, ScienRep.6.26018}, and used to experimentally demonstrate that a certain process --- a quantum-switch~\cite{PhysRevA.86.040301} --- is causally non-separable~\cite{Rubinoe1602589, PhysRevLett.121.090503}.  
In the quantum-switch, a qubit is transmitted between two parties, and the order in which the parties receive and act on it is entangled with a second system.
This can result in a scenario in which operations are applied on the system in a quantum superposition of different temporal orders. 
%This leads to a genuinely indefinite causal structure.
The existence of such a superposition has been experimentally demonstrated~\cite{Rubinoe1602589, PhysRevLett.121.090503, NatCommun.6}.
However, the certification of this ``indefiniteness'' of temporal orders was theory-dependent, requiring the assumption that the system under investigation and the applied operations were described by quantum theory. In more detail, Ref.~\cite{Rubinoe1602589, PhysRevLett.121.090503} reported the measurement of a value for a causal witness that could not be explained by any model making the following three assumptions: (a) there was a definite causal order between the parties, (b) each party acted only once, and (c) their operations are described by quantum theory. %Under these conditions, it was concluded that the causal order was indefinite in the experiment.
Nevertheless, the results of these experiments could potentially have also been explained in accordance with hypotheses (a) and (b) within a different theory (i.e., outside the quantum theory). Thus, the nature of indefinite causal order has not yet been probed without the use of quantum formalism to describe it.

%{\color{red} Causal inequalities are theory independent. They only required measuring probabilities.  But the quantum-switch satisfies causal inequalities.}

In addition to theory-dependent causal witnesses, there are also device-independent ways of certifying indefinite causal order via ``causal inequalities''~\cite{NatComm.3.012316, 1367-2630-18-1-013008}.
These inequalities only require measuring the probabilities of outcomes for different parties in the process under consideration without knowledge of the internal functioning of the devices.
Any probabilities that show signalling in only one direction --- which can be interpreted as an influence from the past to the future ---, or that is a convex mixture of processes which allow signalling only in one direction (from $A$ to $B$ or from $B$ to $A$), satisfy causal inequalities.
Nevertheless, it can be shown that the quantum-switch satisfies all such causal inequalities [see Refs.~\cite{NewJournPhys.17.102001, Oreshkov_2016NJP} or the Suppl. Information for details], and, currently, it is not known whether or not it is possible to realize a process which violates a causal inequality.
The question then arises if it is at all possible to prove the existence of an indefinite causal order in a manner that does not rely on the quantum description of the experiment.
This is a relevant question, since such an experimental verification of indefinite causality would show that this is \textit{not} a feature of a particular theory, %but a fundamental property of any theory used to describe nature
but a \textit{fundamental property} of a whole class of theories which can be used to describe nature. %\giulia{\sout{The situation is analogous to that of other no-go theorems (e.g., Bell's theorem), which show that, due to a contradiction with experimental observations, a specific set of hypotheses cannot be true. These hypotheses are formulated to address fundamental properties of nature (such as `definite causality' in the case of the present no-go theorem, or `Bell's locality' in the case of Bell's theorem), and thus their violation can reveal what is by principle impossible in nature.}}

In this work, we answer the above question affirmatively by presenting an experimental verification of Bell's theorem for temporal order, which is formulated outside of the quantum framework.
To this end, we generalize a Bell inequality for temporal order~\cite{mag}, and then experimentally violate it.
The Bell inequality is shown to be fulfilled in a class of so-called ``\textit{generalized probabilistic theories}'' in which the states and the laboratory operations are local, and the operations are applied in a definite order.
The experimental violation of the Bell inequality presented here demonstrates, independent of quantum formalism, that there exist correlations in nature which are incompatible with a class of theories assuming the order of events as locally pre-defined.

Finally, we notice that, while our inequality is valid for a class of generalized probabilistic theories, it does depend on the internal functionality of experimental devices, and in this sense it does not have the same ``device-independent'' status as the original Bell's theorem. Thus, our work provides a proof of indefinite causality which is weaker than a violation of causal inequalities (i.e., a device-independent proof), but stronger than a measurement of a causal witnesses (a theory-dependent proof). It remains an open question, even from a theoretical viewpoint, whether it is possible to provide stronger evidence of indefinite causality than a violation of the present Bell inequality for temporal order.

\section{No-go theorem for definite temporal order}

We now introduce a \textit{no-go theorem} for definite temporal order that applies to a class of generalized probabilistic theories (GPTs) in which the order of local events is assumed to be pre-defined. 
GPTs are a general framework that specifies a set of operations which can be applied on physical systems, assigns probabilities to experimental outcomes~\cite{GPT1,GPT2,GPT3,1367-2630-13-6-063001}, and which encompasses all operational theories -- including classical probability theory and quantum theory as special cases. The no-go theorem which we present here was previously derived in the context of gravity ~\cite{mag}. Our derivation uses an assumption about the initial state of the systems which is weaker than that in Ref.~\cite{mag} (we consider Bell-local states rather than separable states, which are a subset of Bell-local states), and a  different notion of locality. (The relation between the assumptions and implications of the current work and those of Ref.~\cite{mag} are analyzed in Appendix~\ref{Met-Sec:ProofTheorem}-\ref{Met-Sec:RelToTheory}.)%generalizes the one from Ref. \cite{mag} as ours is not based on the use of a separable state, but rather on a Bell-local state (which is a weaker assumption, as separable states are a subset of Bell-local states).
% it does not rely on the existence of the tensor product structure in GPTs.

%{\color{red} Now define and a Bell inequality for causal order, from the assumptions I, II and III. 
%We also require the standard Bell assumptions (of locality, free choice and fair sampling).
%The fact that our experiment violates this not only shows us that QM is incompatible with a definite causal order, but that and GPT is.}
%Building on Ref. \cite{mag}, we will now define a Bell inequality for causal orders. 
%In order to do so, we will first introduce some basic concepts and notation.

We first define what we mean by a causal order in a GPT. Consider a system in the state $\omega \in \Omega$ of a GPT state space $\Omega$ and imagine two parties, Alice and Bob, who perform some operations on this state. For example, suppose that the operation in Alice's laboratory is given by a transformation $\mathcal{A}$ and that in Bob's laboratory is given by a transformation $\mathcal{B}$. Alice's and Bob's operations are said to undergo a process that is ``causally separable'' in GPTs whenever Alice's operation happens before or simultaneously to Bob's ($\mathcal{A} \preceq \mathcal{B}$), Bob's operation happens before or simultaneously to Alice's ($\mathcal{B} \preceq \mathcal{A}$), or there is a convex mixture of these two cases:
\begin{equation}
\label{eqn:convex_mixture}
\mathcal{S}(\omega)=\zeta \cdot \mathcal{B}\bigl(\mathcal{A}(\omega)\bigr) +(1-\zeta) \cdot  \mathcal{A}\bigl(\mathcal{B}(\omega)\bigr),
\end{equation}
where $0 \leqslant \zeta \leqslant 1$ is the probability with which one or the other order is chosen and $\mathcal{Y}\bigl(\mathcal{X}(\cdot)\bigr)$ is a composition of operations $\mathcal{X}$ and $\mathcal{Y}$.  %[Referring to what follows, I note that, strictly speaking, the quantum-switch is the case of $N=3$. I wonder if we should adapt Eq. (1) for this case.]
(While in the current work we limit our analysis to the case of only $N=2$ parties, an analogue relation can be established for $N>2$ parties, giving rise to a classical mixture of all possible permutations among the $N$ parties, or to  a dynamical causal order, where the causal order between operations may depend on operations performed beforehand~\cite{Abbott2017genuinely}.) If a process cannot be written in the form of Eq.~\eqref{eqn:convex_mixture}, it is called a ``causally non-separable process.''
%Application of this concept to quantum mechanics has shown that not all processes can be described in this form \cite{NatComm.3.012316, PhysRevA.86.040301}.
%This has led to the notion of `\textit{causally separable}' quantum processes (processes that can be described by Eq.~\eqref{eqn:convex_mixture} ), and `\textit{causally non-separable}' quantum processes (processes that cannot be written in the form of Eq.~\eqref{eqn:convex_mixture} ).

%From this definition, it is possible to introduce the concept of the \textit{order between events} (as we have already defined in Eq.~\eqref{eqn:convex_mixture}) and the fundamental idea that, regardless of their order, a causal relation between two events can only be discussed if both events actually occur. %One must also assume the existence of a state and identify the \textit{events} as  transformations of that state.
%These notions, while not being independent of the \textit{device} we are testing, do not require any specific concept from quantum formalism and thus allow us to construct a theory-independent test of the causal orders in our experiment.

Within the GPT framework, we now consider $\omega$ to be a state of the following composite system: one system (the \textit{control} system) governing the order in which the operations $\mathcal{A}$ and $\mathcal{B}$ are applied, and another system (the \textit{target} system) on which the operations are performed. We will further consider that there are two parties, S1 and S2, each possessing one such composite system. No restrictions are applied to the state of the control system (thus, for instance, the composite control state may violate a Bell inequality).
%{\red [STILL NEEDED? Whenever the control system is in a well-defined (classical-like) state, the target system correspondingly propagates along a single wire undergoing a set of transformations defined in a GPT  in a definite order.
%Instead, a transformation on a multiple wire (each representing a single system) generally describes an interaction between the corresponding systems. Then, a `local transformation on multiple wires' corresponds to the special case in which there is no interaction and each system undergoes a separate transformation along its wire.
%In the no-go theorem, we will specify how the `laboratory operations' --- i.e., those that are implemented by Alice ($\mathcal{A}$) and Bob ($\mathcal{B}$) in our experiment --- are represented by transformations in the class of GPTs under consideration.]
%Furthermore, we will consider only experimental arrangements where each given wire contains a {\it fixed set of operations} (i.e., operation $\mathcal{A}$ and operation $\mathcal{B}$ are applied exactly once in each wire, and no other operations occur). Indeed, there must always be more than one event for the `order between events' to be meaningful, and every event should feature in all orders under study.}

In Appendix~\ref{Met-Sec:ProofTheorem} we prove a \textit{no-go theorem}, stating that any two-party system obeying the following three assumptions cannot violate a Bell inequality (below we briefly summarize our theorem, saving the detailed version for Appendix~\ref{Met-Sec:ProofTheorem}).
%{\red [We should discuss what I propose below, as I am not sure about it. In any case, I tried to address comment 1.9. In line of what the referee says "act only on the target system" is ambiguous, because although one operationally acts only on the target system, this action can influence correlations or the state of the control or even the distant target system.]}

%; see also \cite{mag} for a similar proof).
%\vspace{3mm}
%\begin{quote}
%\begin{myframe}
\vspace{2mm}
%\textbf{1. Free-choice assumption (or measurement independence).} 
\begin{enumerate}[label=\Roman*)]
%\item \textbf{\red The laboratory operations $\mathcal{A}$ and $\mathcal{B}$ do not give rise to a violation of Bell inequalities between the target systems.}
\item The initial joint state of the two target subsystems is \textit{Bell-local} (i.e., it satisfies the Bell's ``local-causality" condition for any pair of local measurements).

%\item {\red\textbf{The laboratory operations within each party S$i$ act only on the corresponding target system.}}
\item The laboratory operations preserve Bell-locality of the state of the target subsystems (i.e., when applied to a Bell-local state of the target systems, they produce a Bell-local state.).

\item The order of local operations on the two target subsystems is well-defined.
\end{enumerate}
%We assume that the causal order of both parties can be described by equation~\eqref{eqn:convex_mixture}. 
%{\color{blue}This can be proven in a theory-independent manner, see for example \cite{PhysRevLett.84.2726, PhysRevLett.84.2722} CAN IT?}.
\vspace{2mm}
%\end{myframe}
%\end{quote}
%The assumption 1. means that laboratory operations that act only on the target systems
%\item \textbf{ neither can increase correlations between (distant) target systems belonging to two parties;}\textbf{}
%\item \textbf{nor between a target and control system of a single party}.
%\end{enumerate}
We will briefly comment on assumptions I and II below, and refer to Sec.~\ref{sec:Results} for an in-depth analysis of all three assumptions. 

In our experiment and for the class of GPTs considered, we can demonstrate stronger conditions than I and II. Indeed, the initial joint state of the targets is separable within the GPTs; and the laboratory transformations are local maps on the target systems within the GPTs (i.e., they transform separable states into separable states). Then, I and II follow from these stronger conditions since any separable state is Bell-local.  We could have formulated I and II in terms of separable states of the GPTs, but decided to leave them in the present form, since Bell's locality is a weaker condition than separability, and the formulation does not involve theory-dependent notions.

There are two forms of violation of Assumption II: IIa) The laboratory operations can induce non-local interaction between the target subsystems of parties S$1$ and S$2$. IIb) Within a single party S$i$, $i=1,2$, the laboratory operations may  
``couple'' the control and target system. Such a ``coupling'' could transfer existing non-local correlations between the pair of controls to the pair of targets, thereby enabling a violation of Bell inequalities. We will provide experimental evidence that neither case occurs in our experiment. %cannot increase the non-local correlations between the control and the target subsystems (i.e., they do not `couple' the two subsystems).}

In the next section, we will present a quantum mechanical process that violates this no-go theorem.
Thus, at least one of the assumptions must not hold for this process.
%for which the underlying GPT is assumed to satisfy assumptions I, II and III.
In Sec.~\ref{sec:Results}, we will analyse our experimental data testing a Bell-like inequality to provide evidence in support of assumption I within the framework of GPTs. Consequently, either assumption II does not hold, assumption III does not hold, or both assumptions are invalid. On the basis of the data collected for the quantum-switch of system S1 (or S2) individually, we will show that it is not possible to describe our results by violating only assumption II. Thus, the only viable conclusion is that the order of operations applied on each system $\text{S}i$ is indefinite (i.e., that assumption III is necessarily false).

\section{Entangled quantum-switch}

%The quantum-switch is a two-qubit process that is known to be causally non-separable {\color{blue} REF}.
%See either panel of Fig.~\ref{img:scheme_experiment}.
To understand a single quantum-switch, first imagine two parties, Alice and Bob, who are in two \textit{closed laboratories}, i.e., their only interaction with the external environment is through input and output systems.
Suppose that each of the parties performs an operation on the same qubit (a ``target'' qubit), and that this qubit may be sent first to Alice and then to Bob, or vice versa.
Now, in a quantum-switch, one governs the order of the operations on the target qubit according to the state of a second quantum system, a ``control'' qubit. If the control qubit is placed in a superposition, this establishes a quantum-superposition of the order of the two operations. For instance, if the control qubit is in the state $\ket{0}^c$, the target qubit is sent first to Alice and then to Bob, and vice versa if the control qubit is in the state $\ket{1}^c$. When the control qubit is prepared in the state $\bigl({\ket{0}^c+\ket{1}^c}\bigr) /\ \sqrt{2}$, the resulting process has been shown to be causally non-separable within quantum mechanics~\cite{PhysRevA.88.022318, NewJournPhys.17.102001, PhysRevA.86.040301, Rubinoe1602589}.

Next, consider two quantum-switches (S1 and S2), each containing an Alice and a Bob.
S1 and S2 are prepared in a state where their control qubits are entangled, but their target qubits are in a product state (see Fig.~\ref{img:scheme_experiment}):
\begin{equation}
\ket{0}_1^t\otimes\ket{0}_2^t\otimes\left(\dfrac{\ket{0}_1^c\otimes\ket{0}_2^c - \ket{1}^c_1\otimes\ket{1}^c_2}{\sqrt{2}}\right).
\end{equation}
The superscripts $c$ and $t$ refer to the control and target qubits within one quantum-switch, respectively, while the subscripts $1$ and $2$ refer to quantum-switch S1 and S2.
Since we will attempt to observe a Bell violation with the target qubits, which are in a separable state, this initial condition satisfies assumption I in quantum theory.

\begin{figure}[t]
\centering
\includegraphics[width=\columnwidth]{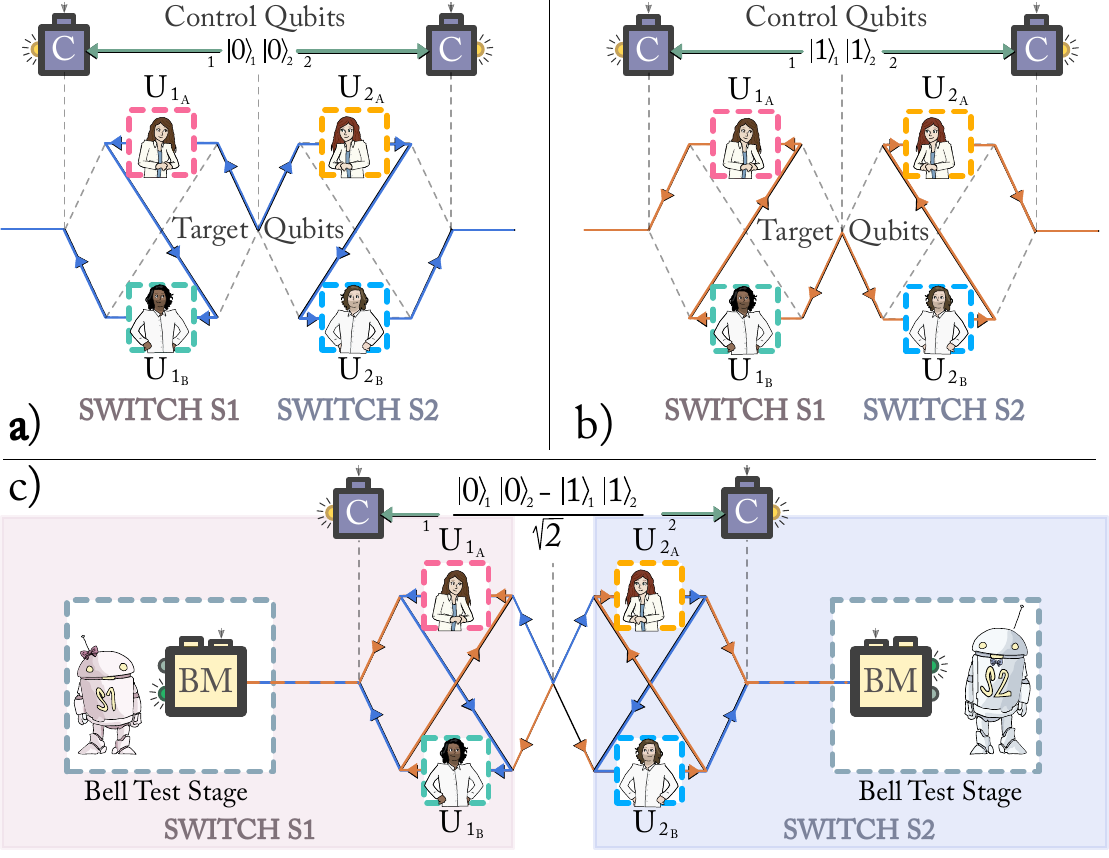}
\captionof{figure}{\footnotesize \textbf{Entangled quantum-switch}. 
Our work is based on two quantum-switches (S1 and S2).
In each quantum-switch, there are two parties, Alice (${U}_{i_\text{A}}$) and Bob (${U}_{i_\text{B}}$). A target qubit is first sent to one party, and then to the other.
The order in which the qubit is sent to the two parties is governed by the state of an additional qubit: if the state of the control qubit is $\ket{0}_i^c$, the target qubit is sent first to Alice and then Bob (Panel \textbf{a}), and vice versa if the control qubit is in the state $\ket{1}_i^c$ (Panel \textbf{b}). 
In our work, we entangle the control qubits (Panel \textbf{c}).  
In this case, the order in which the target qubit in quantum-switch S1 passes through ${U}_{1_\text{A}}$ and ${U}_{1_\text{B}}$ is entangled with the order in which the target qubit in quantum-switch S2 passes through ${U}_{2_\text{A}}$ and ${U}_{2_\text{B}}$.
The control qubits are measured in the basis $\lbrace\ket{+}_i^c, \ket{-}_i^c\rbrace$. 
If the orders inside the two quantum-switches are entangled, it will be possible to violate a Bell inequality by measuring the target qubits after the quantum-switches (BM).
This is possible even if the target qubits start in a separable state and only local operations are applied within each quantum-switch.
}
\label{img:scheme_experiment}
\end{figure}

Given this input state and the action of an individual quantum-switch, it is straightforward to calculate the output of the entangled quantum-switch system
\begin{align}
&\dfrac{1}{\sqrt{2}}\Bigl({U}_{1_\text{B}} {U}_{1_\text{A}}\ket{0}_1^{t}\Bigr)\otimes\ket{0}_1^{c}\otimes\Bigl({U}_{2_\text{B}} {U}_{2_\text{A}}\ket{0}_2^{t}\Bigr)\otimes\ket{0}_2^{c}\\
&- \dfrac{1}{\sqrt{2}}\Bigl({U}_{1_\text{A}} {U}_{1_\text{B}}\ket{0}_1^{t}\Bigr)\otimes\ket{1}_1^{c}\otimes\Bigl({U}_{2_\text{A}} {U}_{2_\text{B}}\ket{0}_2^{t}\Bigr)\otimes\ket{1}_2^{c},\notag
\end{align}
where ${U}_{i_\text{A}}$ and ${U}_{i_\text{B}}$ $(i=1,2)$ are the unitaries performed by the two parties Alice and Bob inside each quantum-switch S$i$.

Next, we measure the two control qubits in the basis $\{\ket{+}\bra{+}, \ket{-}\bra{-}\}$.
If we observe both of the control qubits in the same state (either $\ket{+}_1^c\ket{+}_2^c$ or $\ket{-}_1^c\ket{-}_2^c$), the target qubits will be in the (in general) unnormalised state
\begin{align}
&\dfrac{1}{\sqrt{2}}\bigl({U}_{1_\text{B}} {U}_{1_\text{A}}\ket{0}^{t}_1\otimes {U}_{2_\text{B}} {U}_{2_\text{A}}\ket{0}^{t}_2\notag\\
&\quad- {U}_{1_\text{A}} {U}_{1_\text{B}}\ket{0}^{t}_1\otimes {U}_{2_\text{A}} {U}_{2_\text{B}}\ket{0}^{t}_2 \bigr),
\end{align}
while, if we find the control qubits in orthogonal states (either $\ket{+}_1^c\ket{-}_2^c$ or $\ket{-}_1^c\ket{+}_2^c$), the sign between the two terms in the superposition in the equation above is ``+.''
In general, depending on the choice of the unitaries in the two quantum-switches, the target qubits will be left either in a separable or in an entangled state. In particular, if we choose the gates
\begin{subequations}
\begin{align}\label{eq:gates}
&{U}_{1_\text{A}} = {U}_{2_\text{A}} = \sigma_z, \\
&{U}_{1_\text{B}} = {U}_{2_\text{B}} = \dfrac{\mathbb{1}+i\sigma_x}{\sqrt{2}},
\end{align}
\end{subequations}
where $\sigma_x$ ans $\sigma_z$ are the Pauli operators, the state of the target qubits (upon finding the control qubits in $\ket{\pm}_1^c\ket{\pm}_2^c$) becomes
\begin{equation}\label{eq:entangledState}
\dfrac{1}{\sqrt{2}}\bigl(\ket{l}_1^{t} \ket{l}_2^{t} -\ket{r}_1^{t} \ket{r}_2^{t} \bigr),
\end{equation}
where $\ket{r} = \bigl(\ket{0}-i\ket{1}\bigr) /\ \sqrt{2}$ and $\ket{l} = \bigl(\ket{0}+i\ket{1}\bigr) /\ \sqrt{2}$ (analogously, one gets $\bigl(\ket{l}_1^{t} \ket{l}_2^{t} +\ket{r}_1^{t} \ket{r}_2^{t} \bigr)/{\sqrt{2}}$ when the control qubit is in $\ket{\pm}_1^c\ket{\mp}_2^c$.). This is a \textit{maximally entangled state} and, as a result, one can now violate a Bell inequality on the target qubits. 

Within quantum theory, the entanglement between the targets and the resulting violation of the Bell inequality can be explained in terms of the indefiniteness of the temporal orders in the two quantum-switches. In other words, such entanglement is not ``generated,'' but rather ``transferred'' from the control qubits by means of the indefinite temporal order of the unitaries applied. A related interpretation of the violation in quantum mechanics is in terms of \textit{time-delocalized quantum operations}~\cite{Oreshkov1801.07594v1} and \textit{causal reference frames}~\cite{Gu_rin_2018}, according to which a frame can be chosen such that while Bob's operation acts at a fixed time,  Alice's operation is in a superposition of being implemented before and after Bob's operation, thus resulting in an indefinite causal order between them.

In the class of GPTs considered here, the presence of non-classical correlations can be determined through a violation of a Bell inequality. In our case, the violation of a Bell inequality with the target subsystems implies the violation of the no-go theorem for temporal order, thereby proving that no underlying GPTs where assumptions I, II and III hold can explain the experimental data.
We will experimentally confirm that I holds both in quantum mechanics, and in our class of GPTs (as detailed in Appendix~\ref{Met-Sec:Ass1}). Then, we will show, both within quantum mechanics and in our class of GPTs, that one cannot describe our results if only assumption II is invalid. We will thus conclude that either assumption III is wrong or both assumptions II and III are false, hence proving the presence of indefinite causal order beyond the quantum framework.

%\mz{\textit{[MZ: I don't understand the sentence `Outside the quantum theory....' Is it different `inside quantum theory'? Also`demonstrating outside the quantum formalism that assumptions I and II hold,' is  unclear. Do you mean: demonstrating without assuming validity of QM?]}}
%Moreover, we can certify that the target qubits are entangled using  well-established device-independent tools, by measuring a Bell violation--this our experimental goal.

\vspace{3mm}

\subsection{Experimental scheme}

We create a quantum-switch with entangled control qubits using a photonic set-up. Let us first consider a single quantum-switch. Each quantum-switch applies gates on a target qubit, where the gates' order depends on the state of a control qubit.  Experimentally, we encode the control qubit in a path degree of freedom (DOF), and the target qubit in the polarization DOF of a single photon. 
The photon is initially placed in a superposition of two paths (as explained in Fig.~\ref{img:AliceBob_setup} and Appendix~\ref{Met-Sec:Source}). These paths are labeled $0_1$ and $1_1$ for quantum-switch S1 and $0_2$ and $1_2$ for quantum-switch S2 in Fig.~\ref{img:AliceBob_setup}. The two paths are then routed through a two-loop \textit{Mach-Zehnder interferometer}~\cite{NatCommun.6, Rubinoe1602589}.
The $0_i$ paths lead the photons through a set of gates acting on the polarization DOF in the order ${U}_{i_\text{A}} \preceq {U}_{i_\text{B}}$. 
While the paths $1_i$ guide the photons through the gates in the opposite order ${U}_{i_\text{B}} \preceq {U}_{i_\text{A}}$. 
To generate the maximally entangled state between the target qubits in Eq.~\eqref{eq:entangledState}, we need to implement the non-commuting gates ${U}_{i_\text{A}}=\sigma_z$ and ${U}_{i_\text{B}}=(\mathbb{1} + i \sigma_x)/\sqrt{2}$, which we do with waveplates. In particular, a half-waveplate (HWP) at $0^\circ$ for $\sigma_z$ and a sequence of quarter-waveplate (QWP) and HWP both at $45^\circ$ for $(\mathbb{1} + i \sigma_x)/\sqrt{2}$). After this, the two paths are recombined on a 50/50 beamsplitter (BS) --- which projects the path DOF in the basis $\lbrace\ket{+}\bra{+}, \ket{-}\bra{-}\rbrace$. The path lengths and the relative phases are set by means of a piezo-driven trombone-arm delay line. At the two outputs of each interferometer, QWPs, HWPs and polarizing beam splitters (PBSs) are used to perform arbitrary polarization measurements on the target qubits.

To entangle the two quantum-switches, we first entangle the path DOFs of the two photons.  As explained in Appendix~\ref{Met-Sec:Source}, we generate path-entangled photon pairs that are separable in their polarization DOF:
\begin{align}\label{eq:input}
&\ket{\Phi^{-}}_{1,2}^{\text{path}}\otimes(\ket{H}_1\ket{H}_2)^{\text{polar.}} =\notag\\
&\left(\dfrac{\ket{0}_1\ket{0}_2-\ket{1}_1\ket{1}_2}{\sqrt{2}}\right)^{\text{path}}\otimes(\ket{H}_1\ket{H}_2)^{\text{polar.}}.
\end{align}
Each photon is thus delocalized over two paths. The two photons are then sent to their respective quantum-switches, and, since the control qubits began in an entangled state, the order in which the gates act on the two target qubits becomes entangled.

\begin{figure*}[t]
\centering
\includegraphics[width=\textwidth]{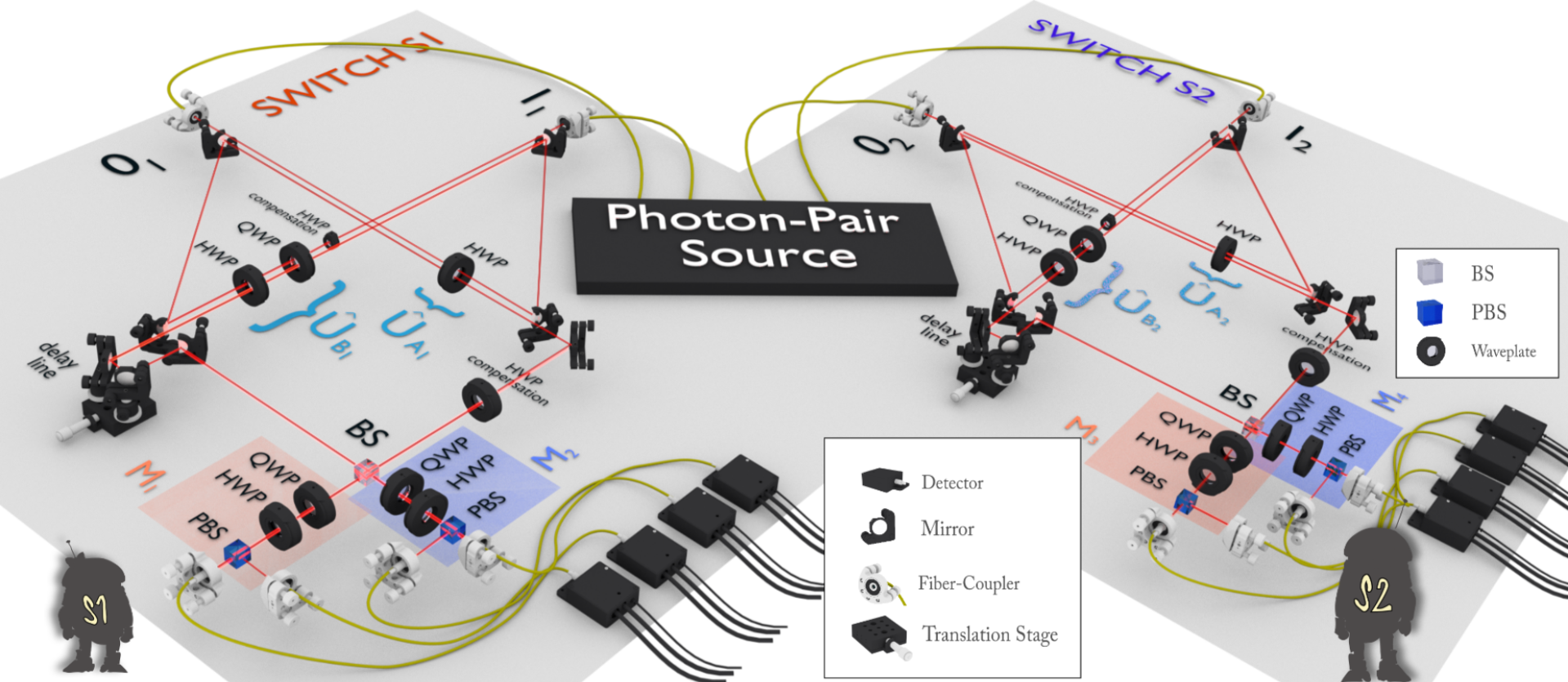}
\captionof{figure}{\footnotesize \textbf{Experimental implementation of an entangled quantum-switch.} 
Each quantum-switch is composed of a two-loop Mach-Zehnder interferometer.  The interferometers start in the photon-pair source, wherein photon 1 and photon 2 are placed in superposition of the paths $0_1$ and $1_1$, and $0_2$ and $1_2$, respectively (see Appendix~\ref{Met-Sec:Source}).
(For simplicity, we have drawn these paths as fibers, however the photons are transmitted via free-space from the source to the experiment.)
These paths are routed such that path $0_i$ sees gate ${U}_{i_\text{A}}$ and then gate ${U}_{i_\text{B}}$, and vice versa for the path $1_i$.  Each gate, acting on the polarization degree of freedom, is made up of waveplates (as described in the main text).
The paths $0_i$ and $1_i$ are then combined on a beam splitter (BS).
In quantum-switch S1 (S2), the photon is detected after the polarization measurement at $M_1$ or $M_2$ ($M_3$ or $M_4$). Both quantum-switches contain two compensation half-waveplate at the beginning and at the end of the reflected arm, so as to compensate for the phase shifts due to the reflection from the polarizing beam splitters composing the source (shown in Fig.~\ref{img:source}).
Together with the BS (which applies a Hadamard gate to the qubit encoded in the path DOF), detecting the photon at $M_1$ or $M_2$ ($M_3$ or $M_4$) projects the path qubit on $\ket{+}$ or $\ket{-}$, respectively.
Furthermore, within each measurement $M_i$, the polarization qubit can be measured in any basis by a combination of a quarter-waveplate (QWP), half-waveplate (HWP), and polarizing beam splitter (PBS). 
}
\label{img:AliceBob_setup}
\end{figure*}

Through this scheme, we engineered a situation wherein the only way entanglement can be transferred from one pair of systems to another is by means of causally non-separable processes. In our experiment, this transfer takes place between different DOFs of photon pairs.  Although it is often easy to transfer the entanglement from one DOF to another, this is typically done with a device that directly couples the two DOFs; e.g., in the case of path-polarization transfer, a PBS could be used. In our experiment, we used an entangled quantum-switch to accomplish this interchange.  Our quantum-switches do not contain any device which directly couples these DOFs (only waveplates, which act solely on the polarization state, and 50/50 BSs, which act solely on the path state).  Rather, here the interchange occurs because the control qubit (the path) governs the order of the application of gates on the target qubit (the polarization). Then, since we begin with an entangled state of the control qubits, this state is transferred to the target qubits via an indefinite order of the application of the gates. In other words, by choosing a specific set of operations, the temporal superposition of the application of these operations is mapped onto a superposition of orthogonal states. As a result, this transfer of entanglement is the signature of an indefinite temporal order.

\vspace{5mm}
\section{Results}
\label{sec:Results}

Our goal is to demonstrate that the order of application of the gates within the two quantum-switches is genuinely indefinite without  assuming that the laboratory operations and the states of the systems are described by quantum theory. We can arrive at this conclusion in three steps. 
We will first show experimental data that violate a Bell inequality. 
From this we can assert that \textit{at least one} of the three assumptions must be false. We will then prove that assumption I is satisfied in our experiment using a class of GPTs. Thus, one of the remaining assumptions (i.e., assumptions II and/or III) must not hold. We will analyse the case in which only assumption II does not hold within the set of GPTs. By acquiring additional data on a single quantum-switch, we will show that such scenario cannot reproduce the results of our experiment. Consequently, the only two possible explanations are that either assumption III does not hold, or that both assumptions II and III are false. In either case, assumption III must be false, and therefore the local operations within the two quantum-switches have been applied in an indefinite temporal order.

%First, we will show that our experimental data violate a Bell inequality {\color{yellow}that is fulfilled by GPTs (WHAT DOES THIS MEAN??)} when one only performs local operations (assumption II) on an initial state that does not violate a Bell inequality (assumption I), and the local operations are applied in a definite order (assumption III). We will then argue that assumptions I and II are satisfied in our experiment using both quantum theory and a class of GPTs.

%In the following section, we will first take some general considerations on the results of our work, in order to provide the reader with a broad-spectrum reading key of our result. We will then proceed to analyze each assumption individually in order to argue whether it should be discarded or maintained.

%We will now briefly discuss each assumption in the context of our experiment and GPTs.

\subsection{Violation of the no-go theorem}

%{\red Having completed the analysis of each individual assumption, we now proceed with the verification of the violation of the no-go theorem.}   
We begin by performing a Bell test between the target states at the output of the apparatus. This allows us to experimentally probe a conjunction of all three assumptions. 
%Let us, hence, proceed with analyzing the state of the output target system. 
%While the violation of a Bell inequality could originate from any of the three assumptions being invalid, the knowledge of our experimental set-up and our additional measurements provide solid grounds that I and II hold not only in quantum theory, but also in the class of GPTs considered. More specifically, local operations are performed at spatially separated regions of a laboratory and no device coupling the target and order systems was present in our apparatus (assumption I). Moreover, all experimental data were compatible with assumption II (i.e., that the initial target system cannot violate any Bell inequality).
We realize the Bell test (more specifically, we measure a Clauser-Horne-Shimony-Holt (CHSH) inequality~\cite{PhysRevLett.23.880}) on the polarization DOF, using four equivalent measurement set-ups (orange and blue boxes in Fig.~\ref{img:AliceBob_setup}).
Since the 50/50 BSs apply a Hadamard gate on the path qubits, we post-select the control qubits in the same state (either $\ket{+}_1^c\ket{+}_2^c$ or $\ket{-}_1^c\ket{-}_2^c$) by grouping the results of $M_1$ with $M_3$ (orange boxes) and $M_2$ with $M_4$ (blue boxes). We obtain $S_{\text{target}} = 2.55\pm 0.08$. This violates the inequality, and thus also the no-go theorem, by almost 7 standard deviations. Therefore, in our class of GPTs, no theory satisfying assumptions I, II and III is compatible with the experimental data.
%Having previously justified assumptions I and II inside and outside the quantum theory, this implies that assumption III does not hold, i.e., that the order of the operations in our experiment is genuinely indefinite.

\subsection{Verification of assumption I}

We now proceed to test the validity of assumption I, which says that the joint target state (shared between system S1 and S2) is Bell-local. We will show experimentally that, within the limits of our experimental precision, the target systems are compatible with a product state in the class of GPTs under consideration, and they are therefore Bell-local.

%We will show that the bipartite state of the target subsystems does not violate a Bell inequality within a class of GPTs by demonstrating that the experimental data are compatible with a product state of a bipartite system in the class.
To test assumption I within a class of GPT, we assume that the set of ``fiducial measurements'' of the class of GPTs contains ``quantum fiducial measurements'' as a subset. Moreover, we assume that  pure states in quantum theory are also pure states in GPTs.
In this sense, quantum theory is embeddable in the GPT.
This is similar to how classical theory can be embedded in quantum theory (i.e., classical theory has one fiducial measurement in the ``computational basis''). In particular, we consider a class of GPTs wherein the state space of a single two-level system is described by a $d$-dimensional Bloch ball~\cite{GPT3, Masanes2011} (i.e., there are three quantum fiducial measurements, and $d-3$ non-quantum fiducial measurements), with $d > 3$ in general.
The methodological advantage of considering such theories is that, in certain cases, the structure of the state can be inferred despite the fact that only quantum measurements can be made (i.e., only measurements in the three-dimensional quantum subspace).
For this class of theories, it was shown that a single system is in a pure state if there exists a measurement for which the system returns a given result with probability one. Similarly, a bipartite system is in a pure product state, if the above statement applies to each individual system. In more detail, in GPTs a state is pure if it cannot be written as a non-trivial mixture of other states. Moreover, a bipartite system is in a product state if, for all local measurements, the probabilities for outcome pairs on a bipartite-state are equal to the product of the two marginal probabilities of each subsystem. Such a state has perfect correlations only for a pair of fiducial measurements, it exhibits no further correlations in any other pair of fiducial measurements, and it cannot violate a Bell inequality~\cite{GPT1, GPT2, GPT3, 1367-2630-13-6-063001}.

%For this class of theories, it was shown that a state is a pure product one\footnote{Here, by `pure product state' in the GPTs we intend what follows. For arbitrary pairs of local measurements, if the probabilities for outcome pairs on a bipartite-state are equal to the product of the two marginal probabilities of each subsystem, then the state is called a product state. Moreover, if these probabilities are one, the state is pure.} if and only if there are perfect correlations for a pair of fiducial measurements for two systems (\textit{Lemma 1} in Ref.~\cite{GPT3}). We hence perform a set of quantum fiducial measurements (defined in Methods - Sec.~\ref{Met-Sec:ProofTheorem}) \cite{GPT1, GPT2, GPT3, 1367-2630-13-6-063001}, and analyse the resulting probabilities directly; the results are presented in Table S1.} %and \ref{tab:sep2}. %The first four columns are the measured joint probabilities, the second four columns are the products of the marginal probabilities.

For our target photon pair we demonstrated that both photons return value $H$ with certainty. This means that, already from a pair of quantum fiducial measurements, one can conclude that, up to experimental imperfections, the state is a pure product state, and therefore it cannot violate a Bell inequality. This  supports the validity of assumption I in the special case of Bloch-vector theories (see Appendix~\ref{Met-Sec:ProofTheorem}). 
In Table S1, we compare the probabilities for outcome pairs on a bipartite-state to the product of the two marginal probabilities of each subsystem.
The excellent agreement between the two probability distributions indicates that the joint target state is indeed a pure product state, and cannot violate a Bell inequality. This proves that assumption I of our no-go theorem holds for the class of GPTs under consideration.
%We demonstrated almost perfect correlations for the pair of fiducial measurements in the $\{ H, V \}$ basis (i.e., among the quantum fiducial measurements), which supports the validity of assumption I in the special case of Bloch-vector theories (see Methods - Sec.~\ref{Met-Sec:ProofTheorem}). 
%If the probabilities for outcome pairs on a bipartite-state are equal to the product of the two marginal probabilities of each subsystem, then the state is a product state {\red in a `quantum subspace', i.e., restricted to the set of quantum measurements (for more details about the possibility that the systems have correlations outside of the `quantum subspace', see \textbf{Methods - Section III.}).} 
%We experimentally performed this characterization for a wide range of measurements; the results are presented in Suppl. Tabs.~\ref{tab:sep1} and \ref{tab:sep2}.  
%The first four columns are the measured joint probabilities, the second four columns are the products of the marginal probabilities. 
%The excellent agreement between the two sets indicates that the target joint state is indeed a product state, and cannot violate a Bell inequality {\red under quantum measurements}.
We quantify to what extent the two distributions agree by calculating the \textit{root-mean-square} (RMS) difference between the two distributions, resulting in an average difference of $0.6 \cdot 10^{- 2} \pm 2.7 \cdot 10^{- 2}$. Although this value is consistent with zero, one could imagine that this small difference is in fact caused by correlations between the two target systems. In Appendix~\ref{Met-Sec:Ass1} we show, however,  that such small correlations can only give rise to a vanishingly small violation of Bell's inequality. The possible level of violation from this amount of potential coupling is insufficient to explain our  experimentally observed violation.
%For pure states and GPTs in which single two-level systems are described by \textit{d}-dimensional Bloch balls \cite{GPT3, Masanes2011}, where \textit{d} is different than three, a stronger conclusion can be drawn: the state is a product state if and only if there are perfect correlations for a pair of fiducial measurements for two systems (\textit{Lemma 1} in Ref. \cite{GPT3}).
%{\red We indeed demonstrate almost perfect correlations for the pair of fiducial measurements in the $\lbrace H, V \rbrace$ basis, to support the validity of assumption I in the special case of Bloch-vector theories and pure states.}
%This proves that assumption I of our no-go theorem {\red holds for the class of GPTs under consideration.}
Therefore, we have confirmed that the joint target system starts in an (approximately) separable state. Additionally, in Appendix~\ref{Met-Sec:Source}, we experimentally show that the joint control system is initially entangled.  
We then send this joint state into our two quantum-switches and perform measurements on the output state. Furthermore, Appendix~\ref{Met-Sec:within_QM} reports the results of the same studies presented above when a quantum-mechanical description of the experimental set-up is assumed.

Having proven, both in quantum theory and within the class of GPTs, that our no-go theorem is violated and that assumption I is justified, we can conclude that either assumption II, or assumption III, or both must be false. We will now consider the case in which only assumption II is false.

\subsection{Verification of assumption II}

The second assumption of our no-go theorem says that the laboratory operations performed in the two quantum-switches cannot transform the joint state of the target systems of S$1$ and S$2$ from a local state to a non-local one. We will again experimentally prove a stronger condition, namely that laboratory operations are local maps on the targets in the considered GPTs. As discussed earlier, this condition can be violated for two reasons.
First, the laboratory operations (IIa) may  induce a direct interaction between two targets.
To avoid this, one could perform the operations with a space-like separation in which case the condition would be guaranteed in any theory obeying relativistic locality. However, in our experiment, we make the (well-justified) device-dependent assumption that the laboratory operations are local transformations within S$1$ and S$2$ in GPTs, since the transformations of the systems take place at spatially separated parts of the optical table. % to exclude the same class of GPTs by violation of Bell's inequalities for temporal order. 
As a consequence, the first form of violation (IIa) of assumption II arguably does not take place in our experiment.

Let us consider now the second, more substantial, way in which assumption II could be violated.
Since we are analysing the case in which only assumption II is false while assumption III holds, we will now study the scenario in which the laboratory operations occur only in a causally-ordered manner.

The second form of violation of assumption II is that the laboratory operations of one party S$i$ could ``couple'' the control and the target systems (IIb). 
Such a coupling would make it possible to transfer non-local correlations from the control systems of the two parties S$1$ and S$2$ to their target systems, and therefore a violation of Bell's inequalities would be possible.
%{\red Within quantum theory}, laboratory operations (implemented using waveplates) {\red in our experiment} act only on the {\red quantum} state of the target system (the polarization DOF), and there is no device that directly couples it to the control DOF (the path), i.e., there is no gate (such as, e.g., a PBS) that could directly `swap' the entanglement.

We can experimentally prove that, within the class of GPTs considered, the laboratory operations transform product states of the control and the target subsystems into product states within each quantum-switch. We will do so using a similar technique as for assumption I.
%We start by preparing the target and control systems of a single quantum-switch in a product state, and we choose a state for the control system such that the order of operations within each quantum-switch is well-defined.
%{\red (Note that, to implement a definite order of operations on the target, we need to know the internal functioning of the experimental device, and hence the theorem is not device-independent; still it is valid for the considered class of GPTs).}
We start by placing bounds on the degree of coupling between the target and the control qubits within a single quantum-switch in the presence of \textit{only one} of the two operations (i.e., only $U_{A}$ or only $U_{B}$) inside the quantum-switch. 
%To this end, we insert only $U_{A}$ in the quantum-switch in first test, and only $U_{B}$ in the second test. 
We perform the full set of quantum fiducial measurements. With this, we show that the joint probabilities of the target and the control subsystems are factorizable into the products of the two marginal probabilities of each subsystem in the case where either only $U_{A}$ or $U_{B}$ is inserted inside a single quantum-switch. From this we can conclude that the joint probabilities must also remain factorizable even when both $U_{A}$ and $U_{B}$ are inserted if the order is well-defined (or a classical mixture of the two well-defined orders)~\cite{Ftn1}. In other words, in our GPTs, the optical elements do not couple the control and the target subsystem in the 
``quantum subspace'' of the GPT state space. Moreover, the marginal probabilities measured on the control and the target subsystem correspond to a pure (product) state.
From this we also conclude that there can be no coupling in the ``non-quantum subspace'' either. %Furthermore, they do not induce on the targets any `exotic' transformations able to transfer correlations from the `non-quantum subspace' into the correlations within the `quantum subspace' (see \textbf{Methods - Section III.} for details).}
We analyse this by performing a set of measurements on the joint control-target system, and by showing that the joint probabilities can be described by the product of the marginal probabilities (see Tables S2-S3, and Appendix~\ref{Met-Sec:Ass2} for more details).
The RMS difference between the two distributions is, on average, $0.02 \pm 0.03$. %(where the first error is a statistical error due to limited photon counts and the second one is a systematic error due to leakage from the PBSs. This measurement is particularly sensitive to this leakage as it can lead to a non-negligible amplitude in the wrong entry of the measurement apparatus, thereby producing an apparent coupling between the two DOFs).
This value is within one standard deviation of zero, confirming that the probability distribution is consistent with that of a product state. As we discussed for assumption 1, this small discrepancy could be caused by correlations between the control and the target systems.
However, as we show in Appendix~\ref{Met-Sec:Ass2}, these correlations are too weak to explain our experimentally observed violation of Bell's inequality.

\subsection{Implications for the temporal order}

%the small discrepancy between the joint probabilities and the product of the marginal probabilities can be due to correlations between the control and the target systems. However, as we show in Appendix~\ref{Met-Sec:Ass2}, these correlations are too weak to explain the observed violation of Bell's inequalities between the two target systems through an alleged control-target coupling in each quantum-switch.}
In the previous subsection we proved that the laboratory operations do not couple the control and target systems when they are applied in a well-defined order. 
%Notice that although this argument is device-dependent (since it relies on the assumption that the device realizes a definite causal order for certain input states of the target) it is theory-independent since the systems are described by GPTs.
This means that, for our experiment, whenever assumption III holds, assumption II must also hold.
Because the statement $a \Rightarrow b$ is logically equivalent to $\text{\textit{not} } b \Rightarrow \text{\textit{not} } a$, this further implies that if assumption II is invalid, then assumption III must also be invalid.
% However, the violation of the Bell inequality requires that at least one of the three assumptions of the no-go theorem is wrong, so the case where assumption II and assumption III are both correct is not able to describe our experiment, and therefore is to be discarded.
%This solution would be in contradiction with the requirement that, having violated the Bell inequality, at least one of the three assumptions must be false.
It follows that the only two possible scenarios are that (1) assumption II is wrong, and therefore III is also wrong, or (2) assumption III is false, independently of assumption II. In either case, it is not possible to explain our experimental data unless assumption III is discarded. We thus conclude that the local operations in our experiment were applied in an indefinite order.
%therefore the scenarios in which the local operations were applied in a well-defined order are to be excluded.
(Notice that it is not a logical necessity that every time III holds, II must hold in turn. In general, one could find a system with a defined causal order (III) in which the operations are non-local (not-II). This would be the case, for instance, of circuits where control and target systems undergo entangling gates in a definite causal order (e.g., CNOT). Nevertheless, this was shown not to be the case in our experiment, as operations $U_A$ and $U_B$ were proven not not couple control and target systems when applied individually.)

\section{Discussion}

In this work, we entangled the temporal orders between operations applied by two parties and experimentally showed that the resulting temporal order is indefinite, %Through the violation of a {\red no-go theorem for temporal orders}, we were able to do this in a theory-independent way.
%In particular, we showed that a pair of quantum-switches can act as an interface to transfer an entangled state from one pair of qubits (the control qubits), to the order of operations, and, finally, to another pair of qubits (the target qubits).
%This transfer of entanglement only works if the order of operations is entangled (and thus, indefinite).
by violating a Bell inequality using the joint target system after the quantum-switches. We thus verified that the data collected by entangling temporal orders in the quantum-switches cannot be described by a class of (generalized probabilistic) theories under the assumption that the initial joint target state is Bell-local, the operations on the target states are local, and they have a pre-defined order. This did not require the assumption that the systems and the operations are described by the quantum formalism. 
Clearly, for our demonstration to be loophole-free (as proposed in Ref.~\cite{mag}), the standard Bell loopholes (fair-sampling and locality) would need to be closed.
Further loopholes can arise related to the implementation of the quantum-switch. In fact, it is known that experimental data produced by the quantum-switch can be simulated by a causally-separable process if at least one of the operations (either $A$ or $B$) is performed two or more times. In relation to the experimental implementations of the quantum-switch, this is the so-called ``single-usage loophole.'' Closing this loophole would require an operational verification that each operation in the quantum-switch is performed only once. 
For example, this could include implementation of a ``counter'' that would estimate the number of times an operation is performed, or a process tomography on time-delocalized quantum systems~\cite{Oreshkov1801.07594v1}. 
However, our experiment is immune to the single-usage loophole. Indeed, a multiple usage of operations in a defined causal order on either side of the Bell test cannot be simultaneously shown to be local and violate the Bell inequality for temporal order.
%However, our experiment is immune to the single-usage loophole as even a multiple usage of local operations on either side of the Bell test cannot result in a violation of Bell's inequality provided that the operations are performed in a definite causal order.
%(for example, {\red one could think of methods to implement a `counter' to ensure that each operation has been performed only once}, {\red although attributing one-usage to each of the operations is already argued to be compatible with the quantum description of the process \cite{Oreshkov1801.07594v1}}).
%This is a relevant line of research, that has ties to several foundational issues in quantum mechanics.
%The high noise resistance of the process here shown is illustrated by the large number of standard deviations through which the process is achieved.

%Our work represents the first experimental demonstration of the existence of causally unordered quantum systems that does not rely on any underlying theory. Furthermore, we showed that the quantum-switch can be used as an interface to coherently exchange quantum information, which may pave the way for further applications of the quantum-switch \cite{PhysRevA.88.022318, PhysRevLett.113.250402,PhysRevA.92.052326, PhysRevLett.117.100502}.

%To make our experiment theory independent we violated a Bell inequality. 
All previous studies involving quantum processes with indefinite temporal order achieved their goal by superposing the order of operations, rather than entangling them.
The first proposal to entangle the temporal order was made only recently~\cite{mag}.
Here we show that the basis of this theoretical concept is in fact experimentally accessible.
Moreover, we exploit this resource as a new means to validate indefinite causal structures. Techniques to characterize these structures are becoming increasingly relevant, as it is known that these processes can lead to logarithmic~\cite{MartinRenner:2021} advantages in query complexity, and exponential advantages in quantum communication tasks~\cite{PhysRevA.88.022318, PhysRevLett.113.250402,PhysRevA.92.052326, PhysRevLett.117.100502}.

\vspace{2mm}
\textbf{Acknowledgments:}
We thank C. Branciard, F. Costa, B. Daki\'c, M. Jacquet, A. Moqanaki and T. Str\"omberg for useful discussions. \textbf{Funding:} acknowledges financial support from the Royal Society through the Newton International Fellowship No.\ NIF$\backslash \text{R1}\backslash$202512. L.A.R. acknowledges support from the Templeton World Charity Foundation (fellowship no. TWCF0194). M.A. acknowledges support from the
Excellence Initiative of the German Federal and State Governments (Grant
ZUK 81). M.Z.~acknowledges support through an ARC DECRA grant DE180101443, and ARC Centre EQuS (CE170100009). \v C.B. acknowledges support from the John Templeton Foundation, Foundational Questions Institute (FQXi), Austrian Science Fund (FWF) through BeyondC (F7103-N38), the projects no. I-2562-N27 and I-2906, as well as support from the European Commission via Testing  the  Large-Scale  Limit  of  Quantum  Mechanics (TEQ) (No. 766900) project. \v C.B. and P.W. acknowledge support from the Austrian Science Fund (FWF) via Doctoral Programme CoQuS (no. W1210-4), and the research platform TURIS. P.W. also acknowledges support from the Austrian Science Fund (FWF) through BeyondC (F7113-N38), the GIPSS (P30817-N36) and NaMuG (P30067-N36), United States Air Force Office of Scientific Research via QAT4SECOMP (FA2386-17-1-4011) and Red Bull GmbH. \textbf{Author contributions:} G.R., L.A.R., \v C.B. and P.W. designed the experiment. G.R. built the set-up and carried out data collection. G.R. and L.A.R. performed data analysis. F.M. built the single-photon source. M.Z., M.A. and \v C.B. developed the theoretical idea. L.A.R., P.W. and \v C.B. supervised the project. All authors contributed to writing the paper.
%\textbf{Competing interests:} The authors declare that they have no competing interests.
\textbf{Data and materials availability:} All data needed to evaluate the conclusions in the paper are present in the paper and/or the Supplementary Information. Additional data related to this paper may be requested from the authors.

%\section{Materials and Methods}
\appendix
\section{Proof of no-go theorem for temporal order}
\label{Met-Sec:ProofTheorem}

%{\blue I think that this paragraph should appear as first in the section.}
All previous experimental studies of causally non-separable processes~\cite{NatCommun.6, Rubinoe1602589} were dependent on the validity of the quantum theory (i.e., they were \textit{theory-dependent}), and all known physically realizable processes satisfy all causal inequalities (see the Suppl. Information)~\cite{NewJournPhys.17.102001, Oreshkov_2016NJP}.
The latter means that experimental data taken from a given causally non-separable quantum process could still be understood as arising in causal manner, for example from a process with a definite causal order in an underlying generalized probabilistic theory (GPT). Therefore, it is unknown whether a fully theory-independent experimental proof of indefinite causal order is possible. 
%While this is also true for the current work, 

In our current work, we relate a violation of a Bell inequality to the violation of a no-go theorem for temporal order, as proposed in Ref.~\cite{mag}. This results in a proof of causal indefiniteness outside of the quantum framework as it is valid for a large class of generalized probabilistic theories.
In this section, we provide a rigorous introduction to such no-go theorem for temporal order.

%{\blue [I provide below elements of a GPT theory since referee 1 asks explicitly for it in comment 1.6. Although not all elements might appear necessary for the present discussion, their presence here make the paper self-contained and prepare the basis for the claims. Maybe some of the sentences could be moved into the main text where we introduce GPTs.]}

We will begin by giving a brief introduction to the basic elements of GPTs which are necessary for our no-go theorem. A more detailed discussion of the GPT framework can be found in Ref.~\cite{GPT2, 1367-2630-13-6-063001, PhysRevA.75.032304}.

In a GPT, a system is described by a state $\omega$ that specifies outcome probabilities for all measurements that can be performed on it. A complete representation of the state is given by specifying the outcome probabilities of a so-called ``fiducial set.'' The smallest such set defines the number $d$ of degrees of freedom of the system. We restrict our consideration here to binary systems that have two perfectly distinguishable states and no more. For example, the fiducial set for a two-level system in quantum theory consists of the (three) probability outcomes of spin projections along $x$, $y$ and $z$. The state space is a compact and convex set $\Omega$ embedded in a vector space. The extremal states of  $\Omega$ that cannot be decomposed as a convex mixture of other states are called ``pure states.''  An effect $e$ is defined as a linear functional on $\Omega$ that maps each state onto a probability, i.e., $e : \Omega \rightarrow [0, 1]$, where $e(\omega)$ is the probability to obtain an outcome on the state $\omega$.  The linearity is required to preserve the convex structure of the state space.  

A transformation $U$ is a linear map from a state to a state, i.e., $U: \Omega \rightarrow \Omega$. The transformation is linear for the same reason that probabilities have to be linear maps of states. The sequence of transformations $U_1$, ... , $U_n$, in which transformation  $U_1$ ``precedes'' transformation $U_2$, which ``precedes'' $U_3$, \textit{etc.}, is represented by a composition of maps: $U_n \circ ... \circ  U_1$. This defines a {\em definite order of transformations}, which we denote as $U_1 \preceq ... \preceq U_n$.

%{\blue [The next old paragraph addresses partially comment 3.6, but the third referee obviously requires further clarifications on the connection between this work and  Ref. [16].  I have  extended  the text at some places to respect the referee's requirement. Still I do not think that the paragraph is rightly placed in the paper.  This section should remain formal with a necessary interpretation of GPTs, but not about the relation to paper [16].] }

We will now introduce a generalization of the no-go theorem for temporal order, which was originally proposed in Ref.~\cite{mag}.

%Bell's theorem for temporal order was introduced in Ref. \cite{mag}. There it was shown that a superposition of massive objects can effectively lead to `entanglement' in the temporal order between local operations, enabling the violation of a new type of inequality. The resource for the violation in this proposal is a `non-classical space-time' created by macroscopic superposition of large masses. Unfortunately, the physical demands of the proposal make that experiment infeasible. 
%However, quantum control of the order of events can also be achieved without the use of gravitational interaction. This can be done, for example, in an extended quantum circuit model, wherein the order of applied quantum gates is coherently controlled by an ancillary system (the quantum-switch).  
%In the gravitational scheme, the spatio-temporal distance of any pair of events in a space-time region is influenced by a superposition state of the mass, whereas in the linear optical implementation, quantum gates are directly applied to the system (e.g., photons) in an indefinite order.

%{\blue I continue below with the "formal" part on GPTs.  I have modified the previous text.}

In the framework of a GPT, the state of a composite system shared between two parties S1 and S2 is given by $\omega_{1,2} \in \Omega_{1,2}$, where $\Omega_{1,2}$ is the state space of a composite system. The state of a composite system is given by a multiplet consisting of the {\em local states} $\omega_{1} \in \Omega_1$ and $\omega_{2} \in \Omega_2$ of individual systems, the \textit{correlation tensor} $\hat{T}$ and a potential \textit{global parameter} $\xi$~\cite{GPT1, GPT2, GPT3, 1367-2630-13-6-063001}:
\begin{equation}
\omega_{1,2}=\omega_{1,2}(\omega_{1},\omega_{2},\hat{T},\xi).
\label{eqn:omega_12}
\end{equation} The fact that subsystems are themselves systems implies that each has a well-defined reduced state $\omega_1$, $\omega_2$ which does not depend on  which  transformations  and  measurements  are  performed  on  the  other  subsystem; this is often referred to as ``no-signaling.'' We also assume that transformations  and  measurements performed on subsystems commute with each other, so that \textit{one} correlation tensor is enough to describe correlations between them. If this were not the case, we would need to introduce \textit{two} correlation tensors, one when S1 applies operations before S2, and the other when S2 performs operations before S1. Finally, the states in GPT need not to satisfy the \textit{local tomography} condition (stating that reduced states and correlation tensor completely describe the systems' state), but may include a global parameter $\xi$.

For the present case of binary systems, the components of the state in Eq.~\eqref{eqn:omega_12} are given by
\begin{subequations}
\begin{align}
\omega^{(i)}_1 &= p^{(i)} (o_1=1) - p^{(i)}(o_1=-1) ,\\
\omega^{(j)}_2 &= p^{(j)} (o_2=1) - p^{(j)}(o_2=-1) ,\\
T^{(i,j)} &= p^{(i,j)} (o_1 \, o_2=1) - p^{(i,j)}(o_1 \, o_2=-1) ,
\end{align}\end{subequations}
where $i,j=1,...,d$. Here, for example, $p^{(i)} (o_1=1)$ is the probability to obtain the outcome $o_1 = 1$ when the $i$-th measurement is performed on the first subsystem, and $p^{(i,j)}(o_1 \, o_2 = 1)$ is the joint probability to obtain correlated results (i.e., either $o_1 = o_2 = +1$ or $o_1 = o_2 = - 1$) when the $i$-th measurement is performed on the first subsystem and the $j$-th measurement on
the second one. 

An effect $e_{12}$ that maps a state onto a probability for a pair of \textit{local} measurements is given
by $e_{12} = e_{12} (r_1,r_2,r_1r^T_2)$, where $r_i$ is the effect on the state of \textit{i}-th system, and $r^T$ denotes transposition of $r$. (Note that the global parameter does not contribute to the probability for a pair of local measurements). The probability to obtain the effect $e_{12}$ when the system is prepared in the state $\omega_{12}$ is given by 
\begin{align}
p(e_{12}|\omega_{12}) =& \frac{1}{4} \Bigl[1 + (\omega_1 \cdot r_1) \notag \\
&+  (\omega_2 \cdot r_2) + (r_2 \cdot \hat{T} r_1) \Bigr], 
\end{align}
where $(x \cdot y)$ is the Euclidean scalar product between two $d$-dimensional real vectors $x$ and $y$. 

The product state is represented by $\omega_p= \omega_p(\eta_1, \eta_2, \eta_1\eta^T_2, \xi_p)$, where the correlation tensor is of a product form. If we perform a pair of local measurements on the arbitrary product state, the outcome probability factorizes into the product of the local outcome probabilities.

We next introduce a pair of {\em local (reversible) transformations} $(U_1,U_2) : \Omega_{12}  \rightarrow \Omega_{12}$  as a linear map from the space of states of a composite system to itself:
\begin{equation}
(U_1,U_2)(\omega_{12} )= (U_1\omega_1,U_2\omega_2,U_1 \hat{T} U^T_2,\xi'),
\label{transfeq}
\end{equation}
where the global parameter $\xi'$ is, in general, changed under the transformations $(U_1,U_2)$. Since testing our Bell inequality involves only local transformations and measurements, it is sufficient to specify effects for those measurements.

%Suppose, however, that the state $\omega_{1,2}$ is constituted by two substates, one governing the order of the operations which each party performs (\textit{control} state) and the other being the system on which the operations are actually applied (\textit{target} state). Here, the subscript $O$ indicates the subsystem which acts as a control and the subscript $T$ that acting as a target system.

%Furthermore, we define the notion of \textit{order of transformations} (both for local and non-local transformations). We say that transformation ${U}_{\text{A}}$ `precedes' transformation ${U}_{\text{B}}$, which `precedes' ${U}_{\text{C}}$, \textit{etc.}, when the composition of the transformations is written in the order: 

%\begin{equation}
%\dots \, {U}_{\text{C}} \, {U}_{\text{B}} \, {U}_{\text{A}}.
%\end{equation}

In our experiment,  $\omega_1$ and $\omega_2$ themselves are states of composite systems each consisting of a ``control'' and a ``target'' subsystem. Thus, the entire system under investigation consists of four subsystems, a control and a target subsystems of S$1$ and a control and a target subsystems of S$2$. The overall state is 
\begin{align}
\omega_{1,2,3,4}=\omega_{1,2,3,4} (&\omega^t_1,\omega^c_1,\omega^t_2,\omega^c_2, ...,\notag\\
&\hat{T}^{ij}, ..., \hat{T}^{ijk}, ..., \hat{T}^{1234}, \Xi),
\end{align}
where $c$ and $t$ refer to the terms ``control'' and ``target'' subsystems, $\hat{T}^{ij}$, $\hat{T}^{ijk}$ and $\hat{T}^{1234}$ are correlation (sub)tensors describing correlations between pairs $\{i,j\}$, triple $\{i,j,k\}$ and quadruple $\{1,2,3,4\}$ of subsystems, respectively, and $\Xi$ is the set of all global parameters. 

%\begin{subequations}
%\begin{align}
%\label{eqn:omegaT_1}
%\omega_1 & \in \Omega_{1}^C \otimes \Omega_{1}^T \\
%\omega_2 & \in \Omega_{2}^C \otimes \Omega_{2}^T,
%\end{align}
%\end{subequations}
%where $C$ and $T$ refer to the terms `control' and `target' subsystem.\footnote{Since local measurements will only all operations will be performed either on subsystem $O$ or subsystem $T$ alone, a potential global parameter of  the composite system can be omitted without loss of generality.}).
%Thus, the joint state $\omega_{1,2}$ of S1 and S2 has the form  of Eq.~\eqref{eqn:omega_12} with the correlation tensor of the form 
%\begin{equation}
%T = \begin{pmatrix}
%\hat{T}^{CC} & \hat{T}^{CT} \\
%\hat{T}^{TC} & \hat{T}^{TT}
%\end{pmatrix},
%\end{equation}
%where $\hat{T}^{ij}$ is correlation sub-tensor describing correlations between $i$, subsystem of S1, and $j$, subsystem of S2.

%Following assumption I, only local operations on the control system $U_i = U_i^C \otimes \mathbb{1}^T_i$ or on the target system $U_i = \mathbb{1}^C_i \otimes U^T_i$ (where $i=1,2$) are considered here and any operation  representing an interaction between the two systems (i.e., of the form $\sum_i c_i U^C_i \otimes U^T_i$) is excluded from this analysis.

The no-go theorem concerns the reduced state of the two target systems as given by
\begin{equation}
\label{eqn:omega12}
\omega^t_{1,2} = \omega^t_{1,2} (\omega^t_1, \omega^t_2, \hat{T}^{tt}, \xi^t),
\end{equation}
where $\omega^t_1$ and $\omega^t_2$ are states of the target subsystems of S1 and S2,  $\hat{T}^{tt}$ is their correlation tensor, and $\xi^t$ is the corresponding global parameter. %Upon local transformation on the control systems of S1 and S2, their target systems remain unchanged. In our experiment this is the operation on the path degrees of freedom, which is induced by the last beam splitters. 
%According to assumtion II the laboratory operations are represented by  local operations on the target systems as given in Eq.~\eqref{transfeq} or a convex mixture therefrom. %Since the amount of violation of Bell's inequalities is obtained by optimization over of local transformations, it is clear that laboratory operations cannot increase violation of Bell's inequalities. 

Leveraging these definitions, we now present three assumptions, which are the fulcrum of our no-go theorem for a definite local causal order.

%\newpage
%\begin{quote}
%\begin{myframe}
\vspace{2mm}
%\textbf{1. Free-choice assumption (or measurement independence).} When investigating the behavior of a process in order to find out whether there is a set of underlying hidden variables $\lambda$ controlling it, some parameters must necessarily be maintained free, so to rule out the possibility that $\lambda$ has controlled the experiment settings.

%{\blue Sorry, the order is now mixed. We should change it either here or in the defintion.}
\textbf{1. The initial joint state of the target system $\omega_{1,2}^t$ is Bell-local (i.e., it satisfies Bell's local-causality condition for all pairs of measurements)}

Suppose that the two observers can each perform a measurement $\mathcal{O}_1$ and $\mathcal{O}_2$, respectively. We label $m_1$ and $m_2$ as arbitrary measurement choices of S1 and S2, and $o_1$ and $o_2$ as the corresponding outcomes. Under these conditions, we suppose that probabilities obtained from any such measurements can be described through a local hidden variables theory (i.e., in Bell's terms, a theory that satisfy ``local causality''), and therefore it is associated to the probability distribution
\begin{align}
\label{eqn:p}
&p(o_1,o_2|m_1,m_2,\omega^t_{1,2}) = \\
&\int \rho(\lambda) \; p(o_1|m_1,\lambda,\omega^t_{1,2}) \; p(o_2|m_2,\lambda,\omega^t_{1,2}) \; d\lambda, \notag
\end{align}
where $\lambda$ is often referred to as a ``hidden variable.'' We implicitly assume the ``freedom of choice'' condition --- the assumption that the choices of the measurement settings are independent of $\lambda$ --- is fulfilled.

\vspace{2mm}
\textbf{2. The laboratory operations are represented by local transformations $U_i^t$ on the target subsystems. If they are applied to a Bell-local state of the targets, the state remains Bell-local.}

This is satisfied by definition, since in the GPTs the action of a transformation  followed by a (fixed) measurement can be understood as another measurement, and the initial state is assumed to satisfy Bell's condition of local causality for {\it all} pairs of local measurements. Therefore, the action of a pair of local transformations cannot produce a non-local Bell state from a local state. 

\vspace{2mm}
\textbf{3. The order of S1's and S2's operations on the target system is well defined.}

Suppose first that the order of application of the local operations performed inside quantum-switch S1 (${U}_{1_A}^t \preceq {U}_{1_B}^t \preceq \dots$) and those performed inside quantum-switch S2 ($\, {U}_{2_A}^t \preceq {U}_{2_B}^t \preceq \dots$) are fixed. 
Since an ordered sequence of local transformations is still a local transformation, if a Bell-local state undergoes such a transformation on S1's and S2's sides, it remains Bell-local. The state remains Bell-local even if the order of operations is chosen with a given probability distribution due to convexity. The mutual order between S1's and S2's operations is irrelevant, since we have assumed the two classes of operations to commute.
\vspace{2mm}
%\end{myframe}
%\end{quote}

\begin{theorem}
\textit{No states, set of transformations and measurements which obey the assumptions I-III can result in violation of a Bell inequality.}
\end{theorem}

\begin{proof}
Following I, suppose that the initial target state $\omega_{1,2}^t$ is Bell-local. This means that Eq.~\eqref{eqn:p} is fulfilled for an arbitrary pair of local measurements.
Because of III, operations in S1's and in S2's laboratories are applied in a definite order, say ${U}_{1_A}^t \preceq {U}_{1_B}^t \preceq \dots$ in S1's side, and ${U}_{2_A}^t \preceq {U}_{2_B}^t \preceq \dots$ in S2's side. The state evolves, therefore, under a composition of the local operations as
\begin{equation*}
\dots \; ({U}_{1_{\text{B}}}^t, {U}_{2_{\text{B}}}^t) \circ ({U}_{1_{\text{A}}}^t, {U}_{2_{\text{A}}}^t)(\omega^t_{1,2}).
\end{equation*}
Let us restrict ourselves to the case of only two transformations per quantum-switch (${U}_{\text{A}}^t$ and ${U}_{\text{B}}^t$). After the pairs of operations are applied in order ${U}_{1_{\text{A}}}^t \preceq {U}_{1_{\text{B}}}^t$ and ${U}_{2_{\text{A}}}^t \preceq {U}_{2_{\text{B}}}^t$ on the two sides, the state becomes
\begin{align}
\omega^t_{1,2}\,' &= ({U}_{1_{\text{B}}}^t, {U}_{2_{\text{B}}}^t) \circ ({U}_{1_{\text{A}}}^t, {U}_{2_{\text{A}}}^t)(\omega^t_{1,2}) \notag\\
&= \bigl({U}_{1_{\text{B}}}^t \circ {U}_{1_{\text{A}}}^t, {U}_{2_{\text{B}}}^t\circ {U}_{2_{\text{A}}}^t\bigr)(\omega^t_{1,2})
\end{align}
which is still local due to I - III. Hence
\begin{align}
&p(o_1,o_2|m_1,m_2,\omega^t_{1,2}\,') =\\
&\int \rho(\lambda) \; p(o_1|m_1,\lambda,\omega^t_{1,2}\,') \cdot p(o_2|m_2,\lambda,\omega^t_{1,2}\,') \; d\lambda . \notag
\end{align}
In general, the order of operations does not need to be fixed, but can be specified probabilistically by a further hidden variable $\nu$, whose different values correspond to different permutations of the order of operations. We obtain
\begin{align}
&p(o_1,o_2|m_1,m_2) =  \iint \rho(\lambda,\nu) \; p(o_1|m_1,\lambda,\omega^{t,\nu}_{1,2}) \notag \\
& \cdot p(o_2|m_2,\lambda,\omega^{t,\nu}_{1,2}) \; d\lambda \; d\nu, 
\end{align}
where $\rho(\lambda,\nu) $ is the joint probability distribution over the two types of variables, and $\omega^{t,\nu}_{1,2}$ is the final state of the target systems upon application of the transformations in the order given by $\nu$.

Thus, we conclude that a local target state subjected to the action of a set of local operations applied in a pre-defined order can by no means lead to the violation of Bell inequalities, even if the order is chosen probabilistically in each run of the experiment. This concludes the proof.
\end{proof}

\section{Device-independency and theory-independency}
\label{subsec:DI-TI}

Causal witnesses, violation of Bell inequalities for temporal order, and violation of causal inequalities build a hierarchy of the notion of indefinite causality. The weakest notion of indefinite causality is that of causal non-separability, which is formulated using quantum theory. A violation of a causal inequality is the strongest notion as it is formulated solely in terms of observable probabilities $p(a,b|x,y)$ without any assumption about the internal function of experimental devices --- it is therefore device-independent. The violation of a Bell inequality for temporal order should be considered, in our view, a stronger proof of indefinite causality than the measurement of a causal witness, but a weaker proof than a violation of a causal inequality. The reason why it is weaker than a causal inequality violation is that, although it too is formulated in terms of the probabilities $p(a,b|x,y,\omega)$, it also involves the notion of state $\omega$ and the assumption how laboratory operations act on it (see Appendix~\ref{Met-Sec:ProofTheorem}) --- this causes the proof to be device-dependent.  However, it can be defined for a class of generalized probabilistic theories, and therefore it does not rely on the quantum formalism. It is thus more general than the notion of a causal witness. Although the quantum-switch violates a weaker notion of causality, shaped for quantum theory, it cannot violate the stronger (device-independent) notion of causal inequalities. The open question addressed in our work is then whether we can still use the quantum-switch to perform a proof of indefinite causal order independent of quantum formalism. The present experimental study provides an affirmative answer to this question.
%We stress once again that the notion of causality examined here refers to the sequence of laboratory operations, not to the operations themselves. In fact, in the present study and the works upon which it is based, it is assumed that the single operation in a local laboratory has a well-defined causal structure with an input and an output.  This could be, for instance, a measurement in which a quantum system enters the device and subsequently produces a result as a classical output.
%To prove this, we have shown that our experimental results are not compatible with a locally defined order of operations.

%Finally, we note that the operations used in the present experiment are unitary, and as such they are not suitable for transmitting and receiving signals. Nevertheless, Alice and Bob could also, in principle, use local operations consisting of measurements and subsequent repreparations of states in the quantum switch. The demonstration of an indefinite causal order of such operations would indicate that the order of signalling may also not be well-defined.

\section{Relation between the present work and Ref.~\cite{mag}}
\label{Met-Sec:RelToTheory}

In Ref.~\cite{mag}, the position of a massive object serves as a ``control'' quantum system and a quantum system (e.g., a photon) that is exchanged between Alice's and Bob's laboratory as a ``target'' system. By putting the massive object in a macroscopic superposition  of two positions, one closer to Alice's and the other closer to Bob's position, one induces a relative time dilation between Alice's and Bob's laboratory. The superposition of massive objects can effectively lead to ``entanglement'' of the temporal order between local operations, enabling the violation of a Bell-type inequality. In the conceptual framework of general relativity, the resource for the violation is a ``non-classical space-time'' created by macroscopic superposition of large masses. In the second-quantized picture, the superposition can be seen as entanglement in the Fock basis, and the scheme enables one to ``swap'' this entanglement to the final entanglement of the target systems. Unfortunately, the physical demands of the proposal make that experiment infeasible. 
However, quantum control of the order of events can also be achieved without the use of gravitational interaction. This can be done, for example, in an extended quantum circuit model, wherein the order of applied quantum gates is coherently controlled by an ancillary system (the quantum-switch).  
The difference between the two scheme is that in the gravitational scheme, the spatio-temporal distance of {\em any} pair of events in a space-time region is influenced by a superposition state of the mass, whereas in the linear optical implementation, only the order of the gates applied on the propagating system (e.g., photons) is indefinite.

A more detailed analysis of the differences and similarities between the gravitational quantum-switch and its photonic counterpart here presented is given in the Suppl. Information - Sec.~III.

\section{Experimental proof of assumption I in GPTs}
\label{Met-Sec:Ass1}

Recall that assumption I says that the initial target states do not violate a Bell inequality. In the notation introduced above, the initial target state is $\omega_{1,2}^t$. Our  demonstration of assumption I presented here is based solely on experimental data, and can be shown to be valid for a class of GPTs. Our goal is to prove that the input state is a product state, and thus it is local. 

Let us denote the probabilities for measurement outcomes as measured on reduced states of the target system of S$1$ and S$2$ as $p(o_1|m_1, \omega^t_{1})$ and $p(o_2|m_2, \omega^t_{2})$, respectively. 
%In a large series of measurements performed, 
If the state is a local product state then the probability for joint outcomes, as measured on the composite system of the two target subsystems in the initial state $\omega^t_{1,2}$, is factorisable, i.e., it can be expressed as
\begin{equation}
p(o_1,o_2|m_1,m_2, \omega^t_{1,2})= p(o_1|m_1, \omega^t_{1}) \cdot p(o_2|m_2, \omega^t_{2}).
\end{equation}
We experimentally performed a large set of measurements on the input target states, and checked for this property.
The measurements we made are tomographically complete in quantum theory. Nevertheless, in a GPT this might not be the case, as a GPT system may have more degrees of freedom than a quantum system. We thus restrict our considerations to a class of GPTs for which we assume that polarization measurements in three unbiased bases for each photon constitute a subset of the full set of ``fiducial measurements''. For example, in the case of GPTs whose systems are described by Bloch vectors of dimension \textit{d}, three components of the vectors correspond to ``quantum fiducial measurements'' of a single system. Similarly, in the GPTs, the correlation tensor is given by $d^2$ elements of which $9$ elements (i.e., $3$ fiducial measurements for the first times $3$ fiducial measurements for the second system) correspond to the ``quantum subspace'' of the correlations that are accessible through quantum measurements.

Our measurements confirm that the joint probabilities for ``quantum fiducial measurements'' are factorized for the two targets. In general, however, it might be possible that within the subset of quantum fiducial measurements for a bipartite system the joint probabilities are factorized into a product of marginal probabilities although the overall state is not a product one. This is because non-zero correlations could exist between non-quantum fiducial measurements. It would then be possible to transfer correlations from the ``non-quantum subspace'' into the correlations within the ``quantum subspace'' by applying some ``exotic'' (i.e., non-quantum) local transformations. Nevertheless, for our class of GPTs, where subsystems are represented by Bloch vectors of general dimension $d$ and pure quantum states are pure states of the GPTs, we know that if both subsystems individually return probability one for some measurement outcome,
%have perfect correlations for a pair of fiducial measurements, 
the state is a pure  product one and it cannot violate assumption I (\textit{Lemma 1} of Ref.~\cite{GPT3}). More precisely, the state would have the form $\omega^t_{1,2} = \omega^t_{1,2} (\omega^t_1, \omega^t_2, \omega^t_1 (\omega^t_2)^T, \xi^t)$, with $\omega^t_1$ and $\omega^t_2$ being in pure states, i.e., $ ||\omega^t_1|| = ||\omega^t_2|| =1$. In our experiment, we obtain outcomes with probability one for a pair of quantum fiducial measurements on the two target subsystems, and therefore the two subsystems cannot exhibit any further correlations within the non-quantum subspace.

Table S1 shows the values of the probabilities $p(o_1,o_2|m_1,m_2, \omega^t_{1,2})$ (which, for brevity, is indicated as $p_{1,2}$ in the Tables) in the first four columns, and the marginal probability products $p(o_1|m_1, \omega^t_{1}) \cdot p(o_2|m_2, \omega^t_{2})$ (denoted as $p_{1} \cdot p_{2}$ in the Table) in the last four columns, with almost perfect correlations in the $\lbrace H, V\rbrace$ basis. Moreover, the joint and the two marginal probabilities are all almost one for the HH outcome, confirming the high purity of the bipartite state. It can be seen that the two sets of probabilities agree well.  More quantitatively, let us define the \textit{root-mean-square} (RMS) distance between the two sets of probabilities as
\begin{equation}
\label{eqn:dist}
\text{d} = \sqrt{\frac{1}{N} \sum_{o_1,o_2}^{} \sum_{m_1,m_2}^{} \Delta p_{o_1,o_2,m_1,m_2}^2},
\end{equation}
with
\begin{align}
&\Delta p_{o_1,o_2,m_1,m_2} = p(o_1,o_2|m_1,m_2, \omega^T_{1,2}) \notag \\
&- p(o_1|m_1, \omega^T_{1}) \cdot p(o_2|m_2, \omega^T_{2}), 
\end{align}
and where $N$ is the number of data points.
Evaluating this over our results, we obtain a RMS distance of $(0.6 \pm 2.7) \cdot 10^{- 2}$, indicating that the two distributions are equal within error.

Although the two target systems are approximately in a product state, the small discrepancy between the two distributions allows for some correlations between the systems. %Recall that the state of the two target systems is given by~\eqref{eqn:omega12} in the GPTs. 
Following the Peres-Horodecki criterion~\cite{PhysRevLett.84.2726, PhysRevLett.84.2722}, the maximal value of the CHSH inequality in quantum mechanics is given in terms of the two largest absolute values of the correlation tensor singular values, say $t_1$ and $t_2$, as $2\sqrt{t_1^2 + t_2^2}$. From the full set of the fiducial measurements in the quantum subspace, we can estimate the maximal possible amount of violation of the CHSH inequality for the two target systems to be  $2\sqrt{t^2_1 + t^2_2}= 2.12 \pm 0.04$. This is more than ten standard deviations lower than the observed violation of the inequality. Thus, this digression from assumption I cannot explain the violation of the inequality.

\vfill

%\begin{table*}[ht]
%\caption{\textbf{Comparison between the two-states probabilities $p(o_1,o_2|m_1,m_2, \omega^t_{1,2})$ and the products of marginal single-state probabilities $p(o_1|m_1, \omega^t_{1}) \cdot p(o_2|m_2, \omega^t_{2})$ for the input target states. - Part II.} The compatibility between the two sets of probabilities shows the separability of the input target state $\omega_{1,2}^t$.}
%\label{tab:sep2}
%\centering
%\begin{tabular}{c||c|c|c|c||c|c|c|c}
%{\footnotesize Measur. Basis} & $p_{1,2}$ & $p_{1,2^{\perp}}$ & $p_{1^{\perp},2}$ & $p_{1^{\perp},2^{\perp}}$ & $p_1 \cdot p_2$ & $p_{1} \cdot p_{2^{\perp}}$ & $p_{1^{\perp}} \cdot p_2$ & $p_{1^{\perp}} \cdot p_{2^{\perp}}$\\
%\midrule
%R, H & 0.65 & 0.02 & 0.33 & 0.01 & 0.64 & 0.02 & 0.33 & 0.01 \\
%R, V & 0.01 & 0.66 & 0.01 & 0.30 & 0.02 & 0.66 & 0.01 & 0.31 \\
%R, A & 0.39 & 0.27 & 0.22 & 0.13 & 0.40 & 0.26 & 0.21 & 0.14 \\
%R, D & 0.29 & 0.39 & 0.12 & 0.19 & 0.28 & 0.40 & 0.13 & 0.19 \\
%R, R & 0.27 & 0.41 & 0.11 & 0.20 & 0.26 & 0.42 & 0.12 & 0.19 \\
%R, L & 0.41 & 0.25 & 0.22 & 0.12 & 0.41 & 0.24 & 0.22 & 0.13 \\
%L, H & 0.32 & 0.01 & 0.63 & 0.04 & 0.32 & 0.02 & 0.63 & 0.03 \\
%L, V & 0.01 & 0.32 & 0.03 & 0.64 & 0.01 & 0.32 & 0.02 & 0.65 \\
%L, A & 0.18 & 0.14 & 0.42 & 0.27 & 0.19 & 0.13 & 0.41 & 0.28 \\
%L, D & 0.14 & 0.21 & 0.23 & 0.41 & 0.13 & 0.22 & 0.24 & 0.40 \\
%L, R & 0.12 & 0.23 & 0.22 & 0.43 & 0.12 & 0.23 & 0.22 & 0.43 \\
%L, L & 0.21 & 0.12 & 0.44 & 0.24 & 0.21 & 0.12 & 0.43 & 0.25
%\end{tabular}
%\end{table*}

\section{Entangled photon source}
\label{Met-Sec:Source}

\begin{figure*}[htb]
\centering
\includegraphics[width=\textwidth]{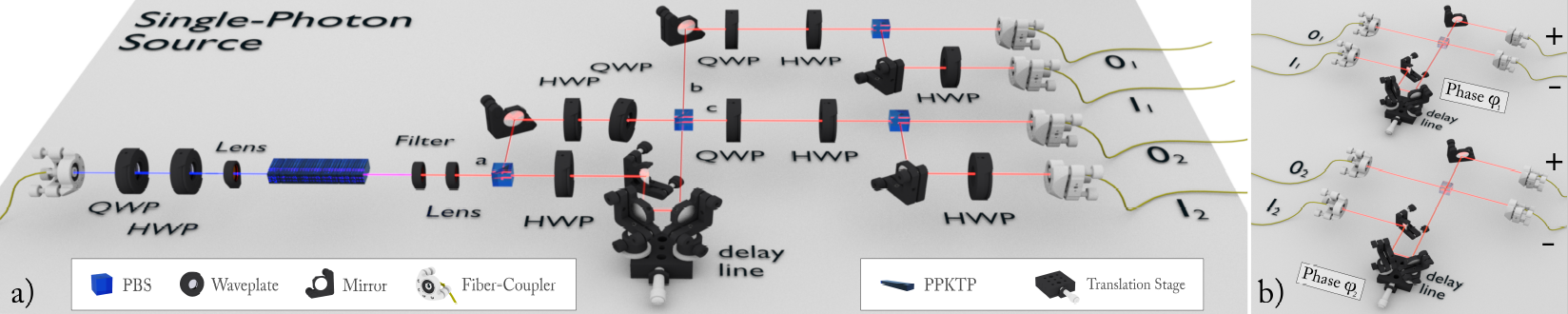}
\captionof{figure}{\footnotesize \textbf{Entangled photon-pair source.} \textbf{a)} \textit{The source} --- The beam from a Toptica DL Pro HP 426 laser is focused on a 30-mm-long PPKTP crystal, phase-matched for degenerate collinear type-II SPDC from $426$ nm to $852$ nm. The phase-matching is finely tuned by controlling the temperature of the crystal with a precision greater than $0.01 K$. The emitted photons have a bandwidth of approximately $0.2$ nm. After the crystal, the residual pump beam is filtered, the photons are then collimated and sent to a set-up to create entanglement by post-selection (as explained in the main text). The entanglement is first produced in polarization and then converted into path using polarizing beam splitters. The source produces $\approx 30.000$ path-entangled photon pairs per second with a pump power of $8$ mW.
\textbf{b)} \textit{Set-up used to measure a Bell Inequality on the path qubits} ---  The two paths composing each qubit are interfered on a beam splitter (BS) projecting each qubit onto a basis on the equator of the Bloch sphere (see main text for more details). 
}
\label{img:source}
\end{figure*}

A \textit{periodically-poled potassium titanyl phosphate} (PPKTP) crystal, phase-matched for collinear type-II \textit{spontaneous parametric down-conversion} (SPDC), converts one photon at $426$ nm into two photons at $852$ nm. The photonic state after the crystal can be approximated to a Fock state of two photons in two orthogonal polarization modes $\vert H, a \rangle \vert V, a \rangle$, where $a$ indicates the common spatial mode of the two photons defined in Fig.~\ref{img:source}. Two PBSs are used to separate and then recombine the two photons. Each photon passes through a HWP set at $\pm 45^{\circ}$. The state after the second PBS is therefore: $\bigl(\ket{H, b} \ket{H, c} - \ket{H, b} \ket{V, b} + \ket{H, c} \ket{V, c} - \ket{V, b} \ket{V, c}\bigr) /\ 2$, where $b$ and $c$ indicate the two output spatial modes of the second PBS. By post-selecting on coincidences, only the part of the state with the photons in two different spatial modes is kept, resulting in the polarization-entangled state $\bigl(\vert H, b\rangle \vert H, c\rangle - \vert V, b \rangle \vert V, c \rangle\bigr) /\ \sqrt{2}$. 
We then use two  PBSs and two HWPs (Fig.~\ref{img:source}) to convert this state into a path-entangled state: $\bigl(\ket{0}_1 \ket{0}_2 - \ket{1}_1\ket{1}_2\bigr) /\ \sqrt{2}$, where the notation is the same as specified in Fig.~\ref{img:source}.
A trombone delay line in between the two PBSs is used to compensate temporal delay between the two photons, and a multi-order QWP in one mode is tilted to compensate for undesired phases between the two components of the final quantum state.
The delay line and the QWP can be also used to modify the final output state in a controllable way. In particular, by unbalancing the two paths by the coherence length of the down-converted photons, the entangled state can be converted into a statistical mixture of the states $\vert 0\rangle_1 \vert 0\rangle_2$ and $\vert 1 \rangle_1 \vert 1 \rangle_2$.
%Additionally, by tilting the multi-order QWP, a phase can be added between these two terms. Furthermore, not maximally entangled states can be produced if the HWPs after the first PBS are set to angles different than $\pm 45^{\circ}$.

For our experiment, both the path and the polarization states of the photon pairs are important.  To characterize the polarization state, we can perform two-qubit polarization state tomography using a QWP, a HWP and a PBS for each photon ({Fig.~\ref{img:source}, Panel \textbf{a}}).
To characterize the path entanglement, we perform a Bell measurement on the path qubits using the apparatus shown in Fig.~\ref{img:source}, Panel \textbf{b}, which essentially consists of one Mach-Zehnder interferometer for each photon.  The phase of the interferometers sets the measurement bases $\lbrace \frac{1}{\sqrt{2}}(\ket{0} + e^{-i\phi_i}\ket{1}), \frac{1}{\sqrt{2}}(\ket{0} - e^{-i\phi_i}\ket{1}) \rbrace$.
Using these two interferometers we can measure all what is required for a CHSH parameter: 
\begin{align}
S=&\Bigl\lvert C(o_1,o_2) + C(o_1',o_2) \notag\\
&+ C(o_1,o_2') - C(o_1',o_2') \Bigr\rvert,
\label{eqn:CHSH}
\end{align}
where
\begin{equation}
C(o_1,o_2)=\dfrac{N_{++}-N_{+-}-N_{-+}+N_{--}}{N_{++}+N_{+-}+N_{-+}+N_{--}}.
\end{equation}
Here, $N_{++}$ is the number of coincidence events between detectors labelled $+$ for each photon in Fig.~\ref{img:source}, Panel \textbf{b}, $N_{+-}$ the number of coincidence events between detectors $+$ and $-$ for each photon, and so on.

Fig.~\ref{img:BellPath} shows the characterization of the entanglement of the joint input control state, %In previous theory-dependent tests \cite{Rubinoe1602589,NatCommun.6}, the temporal order of the operations was made indefinite by placing the state of the control system in a superposition state.
%As we showed above (Section 1.1), one can also achieve this indefiniteness by entangling the two control subsystems. This has the advantage of allowing us to perform a test of the indefinite temporal order {\red outside of the quantum framework}.
%Our results show that one can achieve this indefiniteness by entangling the two control systems. Furthermore, we can use this entanglement to perform theory-independent verification
%Experimentally, the control subsystems are encoded in path DOF of two photons.  
where we verified the initial entanglement by performing a Bell measurement on the joint control system before the quantum-switch, obtaining a CHSH parameter of $2.58\pm0.09$. 

\begin{figure}[ht]
\centering
\includegraphics[width=.9\columnwidth]{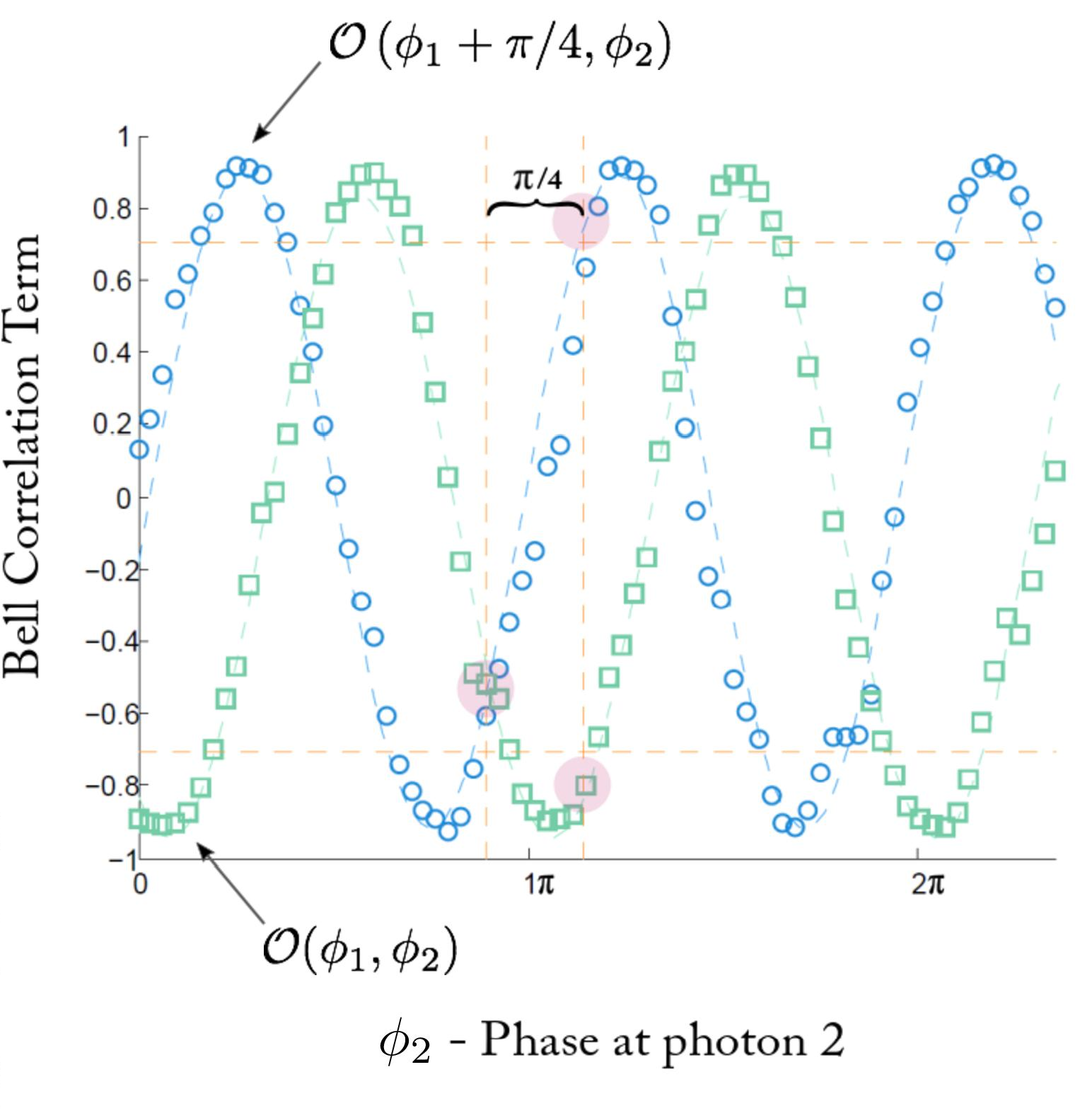}
\captionof{figure}{\footnotesize \textbf{Input control state characterization: Bell measurement on the order qubits.}
Each curve is a measurement of a Bell correlation term $C\bigl(\mathcal{O}_1(\phi_1),\mathcal{O}_2(\phi_2)\bigr)$ on the control qubits, wherein the phase of $\phi_1$ is fixed, and the phase $\phi_2$ is scanned. As described in Eq.~\eqref{eqn:CHSH} of Appendix~\ref{Met-Sec:Source}, we test the Clauser-Horne-Shimony-Holt (CHSH) inequality~\cite{PhysRevLett.23.880} achieving a violation of $2.59 \pm 0.09$.
%The control qubits are path qubits, created in a maximally entangled state.
%The two paths of each qubit are interfered at a beam splitter (BS), projecting each qubit onto the $\ket{0}\pm e^{-i\phi}\ket{0}$ basis, where $\phi$ is the phase of the interferometer (see \textbf{Methods - Section I.} for details).
For the data in the green curve, the phase $\phi_1$ was nominally shifted by $\pi/4$ rad with respect to the blue curve.
The red shaded areas represent the regions where values of $\phi_1$ and $\phi_2$ correspond with those used to construct our CHSH parameter (Eq.~\eqref{eqn:CHSH} of Appendix~\ref{Met-Sec:Source}). In particular, $\mathcal{O}=(\mathcal{O}_1, \mathcal{O}_2)$ where $\mathcal{O}_i(\phi_1,\phi_2)=cos(\phi_i) \, \sigma_x+cos(\phi_i) \,\sigma_z$.
These data confirm that the two photons start in a path-entangled state, and the polarization state is initially separable.
}
\label{img:BellPath}
\end{figure}

%\subsection*{Section VI. Separability between control and target system for ordered application of the gates}

\section{Experimental proof that IIb does not hold in GPTs under the assumption of III}
\label{Met-Sec:Ass2}

In this section, we experimentally prove that the form of violation IIb does not hold for the case in which assumption III holds. Form of violation IIb states that the laboratory operations may couple (i.e., may generate non-local correlations between) the control and the target subsystems, within a given party S$i$, $i=1,2$. 
To test it, we first prepare the control and target subsystems in a tomographically-complete set of product states within quantum theory, i.e., 
\begin{align}
&p(o^c,o^t|m^c,m^t,\omega_i=\omega^c \omega^t)\notag\\
&=p(o^c|m^c,\omega^c) \cdot p(o^t|m^t\omega^t)
\label{eqn:p_omega_i}
\end{align}
for all states $\omega_i = \omega^c \, \omega^t$ from a tomographically-complete set of product states. In the GPTs, Eq.~\eqref{eqn:p_omega_i} shows that, for a product state from the quantum subspace, the probability for a joint outcome factorizes into the product of the probabilities for individual outcomes.
We then set a single quantum-switch to have only one operation inserted, either $U_{i_A}$ or $U_{i_B}$.
%Note that this is a device-dependent step. Nevertheless,  it can be described within GPTs.
We finally verify that, for the full set of preparations, the control and target subsystems are still in a product state after the quantum-switch, when this contains only $U_{i_A}$ or only $U_{i_B}$.
More precisely, we verify that 
\begin{subequations}
\begin{align}
&p(o^c,o^t|m^c,m^t,U_{i_A} \omega_i)\notag\\
&\quad=p(o^c|m^c,U_{i_A}\omega_c) \, p(o^t|m^t, U_{i_A}\omega_t), \\  
&p(o^c,o^t|m^c,m^t,U_{i_B} \omega_i)\notag\\
&\quad=p(o^c|m^c,U_{i_B}\omega_c) \, p(o^t|m^t, U_{i_B}\omega_t), \label{eqn:p_U_i_omega_i}
\end{align}
\end{subequations}
for any state from a complete set of product states, and by linear extension to an arbitrary product state $\omega_i$. We do this using the same technique we used to verify that the target qubits began in an input state (Appendix~\ref{Met-Sec:Ass1}). Finally, we make use of the following property: if neither operation $U_{i_A}$ nor $U_{i_B}$ alone couple the two subsystems, then also a sequence of the two operations cannot couple them as long as they are performed in a definite causal order. This conclusion follows directly from Eqs.~\eqref{eqn:p_omega_i}-\eqref{eqn:p_U_i_omega_i}:
\begin{align}
 &p(o_c,o_t|m_c,m_t,U_{i_B} \circ U_{i_A} \omega_i) =\\
 &p(o_c|m_c,U_{i_B} \circ U_{i_A}\omega_c) \; p(o_t|m_t, U_{i_B} \circ U_{i_A }\omega_t).\notag
\end{align}
Note that even under a multiple usage of $U_{i_A}$ and $U_{i_B}$ there can be no coupling when the operations are performed in a definite causal order. This finalizes the proof of assumption IIb. 

Tables S2-S3 report the values of the probabilities $p(o_c,o_t|m_c,m_t, \omega_{1})$ (which, for brevity, are indicated as $p_{c,t}$ in the Tables) compared with the marginal probability products $p(o_c|m_c, \omega^c_{1}) \cdot p(o_t|m_t, \omega^t_{1})$ (denoted as $p_{c} \cdot p_{t}$ in the Tables).

The tomographically-complete sets of fiducial quantum measurements reported in Tables S2-S3 were performed as follows. In order to vary the input state of the control system among $\ket{+}^c$, $\ket{-}^c$, $\ket{R}^c=\bigl(\ket{0}^c-i\ket{1}^c\bigr)/\sqrt{2}$, and $\ket{L}^c=\bigl(\ket{0}^c+i\ket{1}^c\bigr)/\sqrt{2}$, we set the relative phase between the two trajectories after the first beamsplitter by means of a delay stage mounted on a calibrated piezo-actuator.
Instead, by blocking either path, we prepared $\ket{0}^c$ and $\ket{1}^c$.
Likewise, we measure the path qubit in the following way.
To measure in $\{\ket{+}^c,\ket{-}^c\}$, or $\{\ket{R}^c,\ket{L}^c\}$, we suitably set the relative phase between the two paths before recombining them at the second beamsplitter. This can be done by adding the required phase for state preparation and subtracting the phase for state measurement. Such a phase is then converted into a path delay and sent to the piezo-actuated delay stage.
(We use the same delay stage to both set the phase of the path state, and to measure it in $\{\ket{+}^c,\ket{-}^c\}$, or $\{\ket{R}^c,\ket{L}^c\}$.)
To measure in the $\{\ket{0},\ket{1}\}$ basis, we block either path before the 50/50 beamsplitter, and we then sum the counts from the two paths after the beamsplitter.

The displayed output probabilities $p(o_c,o_t|m_c,m_t, \omega_{1})$ are very close to those corresponding to a product state. This is indicated by the fact that the RMS distance [Eq.~\eqref{eqn:dist}] between these two sets of probabilities (the measured joint probabilities  $p(o_c,o_t|m_c,m_t, \omega_{1})$, and that given by the product $p(o_c|m_c, \omega^c_{1}) \cdot p(o_t|m_t, \omega^t_{1})$) is $(2 \pm 3) \cdot 10^{- 2}$ %(where the first error is a statistical error, while the second is a systematic error)
when only operation $U_{i_A}$ was acting on the system, and $(3 \pm 3) \cdot 10^{- 2}$ when only operation $U_{i_B}$ was present.
Moreover, the displayed output probabilities $p(o_c,o_t|m_c,m_t, \omega_{1})$ are very close to those of a pure (product) state, which means that there cannot be any correlations in the non-quantum subspace. More precisely, the non-vanishing discrepancy between the two probability distributions could be caused by a weak coupling between the control and the target system. Using the same technique as in Appendix~\ref{Met-Sec:Ass1}, we can estimate that the correlations established through this coupling are too weak to violate the CHSH inequality ($2 \sqrt{t^2_1+ t^2_2} = 1.76 \pm 0.04$). The coupling between each pair of the control and the target systems can ``swap'' the correlations from the bipartite state of the control system to that of the target system. However, assuming that the transferred amount of correlations to the target system cannot be larger that the amount produced through the coupling between each pair of the target and control system, we conclude that the coupling cannot result in the target systems violating the Bell's inequality.

This confirms within experimental error, under the hypothesis that assumption III is valid, that the form of violation IIb does not occur in our experiment. Furthermore, this proof holds not only within quantum theory but also for our class of GPTs. In other words, our measurements imply that in both of the two quantum-switches, individually, the laboratory operations do not couple the target and the control subsystems in GPTs when these operations are executed in a definite causal order. %; \textbf{(b)} they do not transfer correlations from the non-quantum to the quantum subspace of the correlations tensor of the two targets on each side. The latter means that the transformations $U_1$ and $U_2$ act on the correlation tensor $\hat{T}^{tt}$ of the two targets in the following way 
%\begin{equation}
%U_1 \hat{T}^{tt} U^T_2=
%\begin{bmatrix}
%  U^{qu}_1
%  & \rvline & \bigzero \\
%\hline
%  \bigzero & \rvline &
%  U^{nqu}_1
%\end{bmatrix}
% \begin{bmatrix} 
%    T^{tt}_{11} & \dots & T^{tt}_{1d} \\
%    \vdots & \ddots & \vdots \\
%    T^{tt}_{d1} & \dots & T^{tt}_{dd} 
%    \end{bmatrix}
%\begin{bmatrix}
%  U^{quT}_2
%  & \rvline & \bigzero \\
%\hline
%  \bigzero & \rvline &
%  U^{nquT}_2
%\end{bmatrix},
%\end{equation}
%where $U^{qu}_1$ and $U^{qu}_2$ represent actions of the two transformations in the quantum subspace, and $U^{nqu}_1$ and $U^{nqu}_2$ in the non-quantum subspace.
%This ensures that the joint probabilities will remain factorizable into a product of marginal probabilities after the bipartite system undergoes the laboratory operations. In the special case of GPTs with Bloch vector representation and the targets in pure product state, assumption II automatically holds, as no local operation can produce from it a state violating a Bell inequality.}
From this experimental test, we thus conclude that in our experiment assumption II cannot be false unless assumption III is also violated.

\vspace{3mm}
\section{Experimental Tests assuming a Quantum Mechanical Description}
\label{Met-Sec:within_QM}

Although our results are not based on the assumption of a quantum description of the experiment, below we report the analysis of the results of our experiment within this description for completeness.

\subsection{Entanglement of the output target state}

In this subsection, we report the results of the polarization-state tomography on the two-qubit output target state after the quantum-switches. 
The resulting density matrix is presented in Fig.~\ref{img:RR-LL}, and it shows a clear presence of entanglement.  The reconstructed state has a fidelity of $0.922 \pm 0.005$ with the ideal one [Eq.~\eqref{eq:entangledState}], and a concurrence of $0.95 \pm 0.01$.

\begin{figure}[htb]
\centering
\includegraphics[width=\columnwidth]{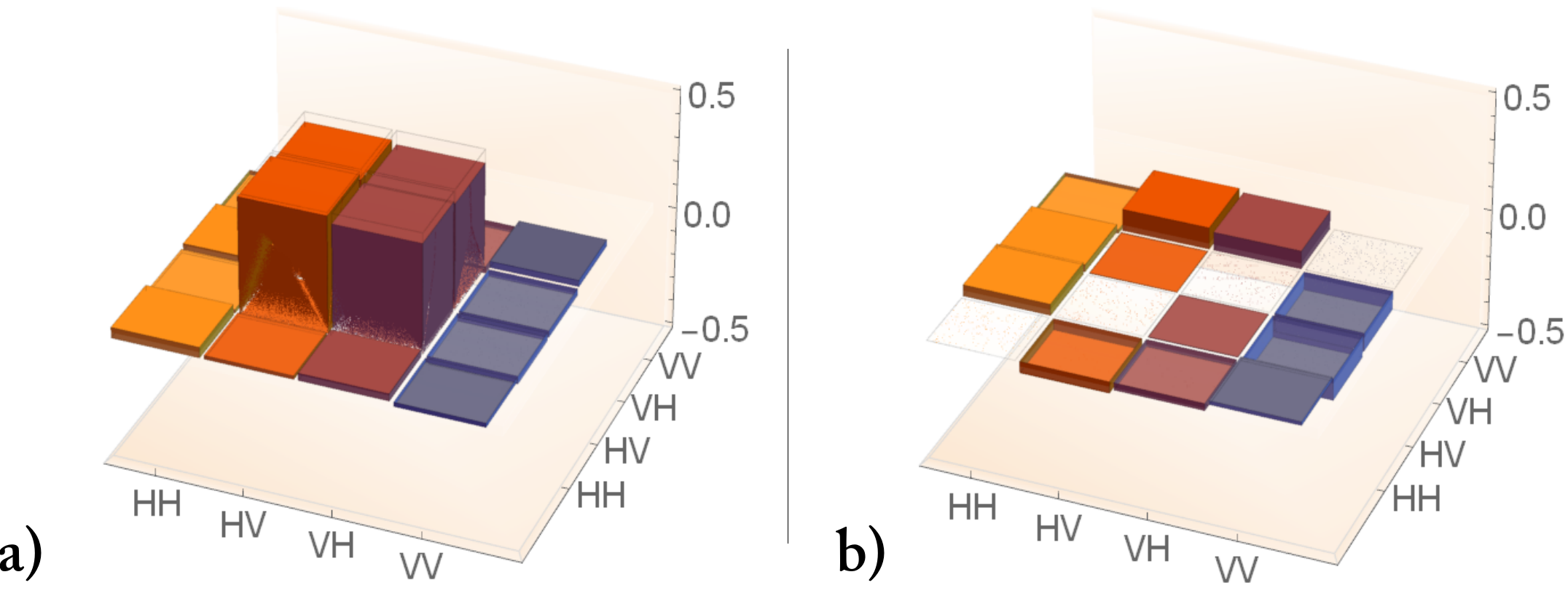}
\captionof{figure}{\footnotesize
\textbf{Output state characterization.} Panels \textbf{a} and \textbf{b} show the real and imaginary parts, respectively, of the two-photon polarization state measured after the photons leave the quantum-switches. For the data shown here, the two control qubits were found to be in the same state (either $\ket{+}_1^c\ket{+}_2^c$ or $\ket{-}_1^c\ket{-}_2^c$).  This state has a fidelity of $0.922 \pm 0.005$ with the target state $(\ket{HV}+\ket{VH})/\sqrt{2}$, and a concurrence of $0.95 \pm 0.01$.
Performing a Bell measurement directly using this state results in a CHSH parameter of $2.55 \pm 0.08$. 
}
\label{img:RR-LL}
\end{figure}

\subsection{Verification of assumption I}

Within quantum theory, one can show the validity of assumption I by demonstrating that the state is separable; this can be done using quantum state tomography, for example. 
To this end, we performed tomography on the target states before the quantum-switches. The resulting density matrix is shown in Fig.~\ref{img:HH}, Panels \textbf{a} and \textbf{b}.  
For our experiment, the target state was nominally prepared in $\ket{HH}$; our measured state has a fidelity of $0.935 \pm 0.004$ with $\ket{HH}$. Furthermore, the concurrence of the estimated state is $0.001 \pm 0.010$, indicating that, within experimental error, the initial target state is separable, in agreement with assumption I. 
The error bars are computed using a Monte Carlo simulation of our experiment; the dominant contribution comes from errors in setting the WPs, and cross-talk in the polarizing BSs.

\begin{figure}[tb]
\centering
\includegraphics[width=\columnwidth]{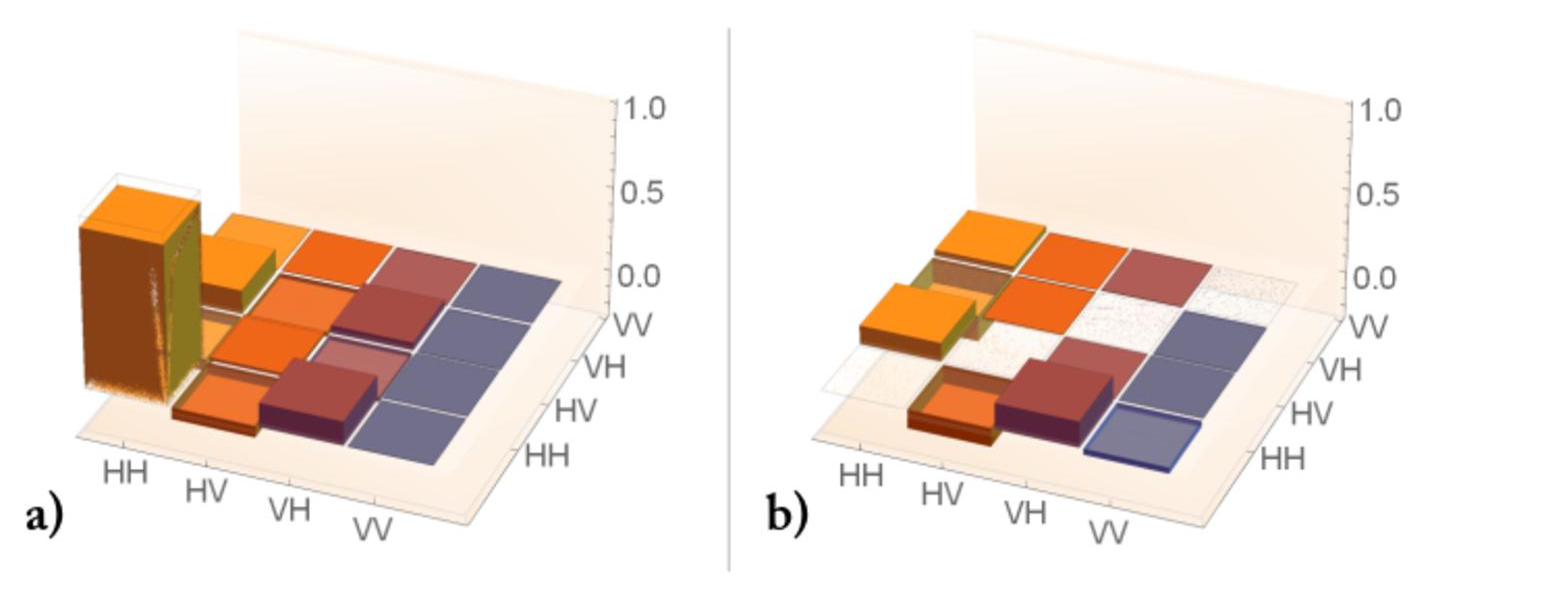}
\captionof{figure}{\footnotesize \textbf{Input control state characterization: State tomography of the target qubits.}
The real (Panel \textbf{a}) and imaginary (Panel \textbf{b}) parts of the two-photon polarization state are measured before the two photons enter the quantum-switches. This state has a fidelity $0.935 \pm 0.004$ with the ideal state $\ket{HH}$, and a concurrence of $0.001 \pm 0.010$.}
\label{img:HH}
\end{figure}

\section{Data analysis}
\label{Met-Sec:Data_Analysis}

In order to convert the coincidence counts into probabilities, we weight each measured count rate by the net detection efficiency of the corresponding detector pair.
We estimate these efficiencies in two parts. 
First, we measure the relative coupling efficiencies between the output ports $M_1$ and $M_2$ of quantum-switch S1, and $M_3$ and $M_4$ of quantum-switch S2.
Then, within each output port, we measure the relative efficiency of the detector in the transmitted port and the reflected port.
We find relative efficiencies between $\approx 0.85$ and $1$.
For more details, see the Methods of our previous work~\cite{Rubinoe1602589}.

The main source of error in our experiment was phase fluctuations.
In the Bell measurement, this dephasing is mainly due to two contributions. (1) Undesired phase-shifts in the interferometer (which we estimated to be about $0.97^{\circ}\pm 0.02^{\circ}$). (2) Fluctuations of the source, which produces time varying phase between the $\ket{HH}$ and $\ket{VV}$ terms.  In our source, we estimate this to be approximately $1.9^{\circ}\pm 0.1^{\circ}$, which is caused by a combination of fluctuations in the pump laser wavelength, and the phase-matching temperature.
We convert these errors into an error in the Bell parameter using Gaussian error propagation.
To calculate the error for the Bell measurements on the polarization qubits after the quantum-switches, we consider the same error sources as above (where now the phase shifts in the measurement interferometer are replaced by phase shifts in the quantum-switches).
However, we also consider errors arising from setting the polarization measurements.
%They results rotated of about $1.5424^{\circ}$ from the ideal states $\ket{H}$, $\ket{D}$, $\ket{R}$ which are used in the theoretical model.
Finally, to estimate the errors in the results extracted from tomography (i.e., fidelity and concurrence), we performed a Monte Carlo simulation considering the phase fluctuations discussed above.

\section{Additional consistency tests: the insertion of noise}
\label{Met-Sec:AddTests}

\begin{figure*}[ht]
\centering
\includegraphics[width=.95\textwidth]{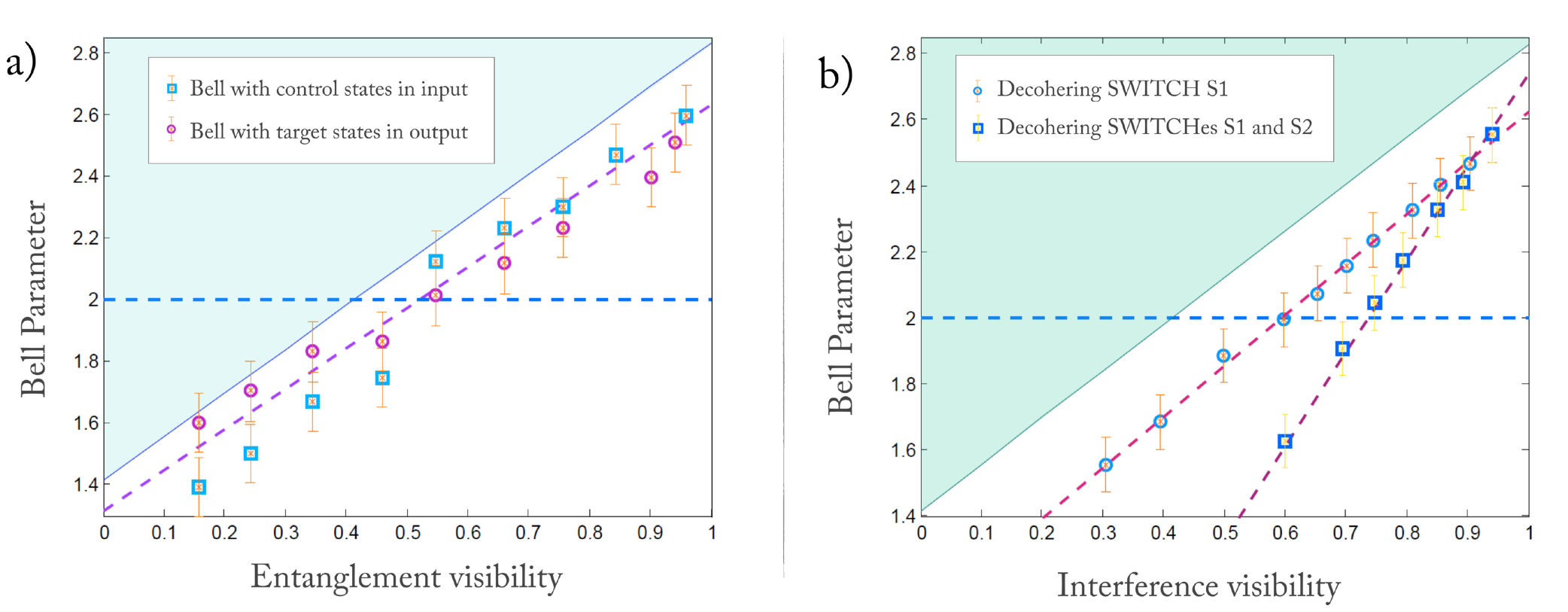}
\captionof{figure}{\footnotesize \textbf{Bell parameter in the presence of various decoherence sources.} 
\textbf{a)} For these data, the initial entanglement of control qubits is decreased passing from the entangled state $\frac{1}{\sqrt{2}}\bigl(\ket{0,0}-\ket{1,1}\bigr)$ to a mixture of $\ket{0,0}$ and $\ket{1,1}$.
We measure the Bell parameter both on the input path qubits (squares) and output polarization qubits (circles) as the source is decohered.  Here, the Bell parameter is plotted versus the visibility of the entangled state in its anti-correlated basis.
The dashed line is a simulation of the experimental results.
\textbf{b)} For these data, the coherence of the superposition of the orders of operations inside the quantum-switches is decreased, leading to a classical mixture of orders.
To control this transition, we decrease the visibility of either only one of the two interferometers (circles), or of both interferometers at the same time (squares). 
Each graph shows the Bell parameter plotted versus the visibility of one interferometer.
The dashed lines are linear fits to the data.
The horizontal dashed blue line, in both plots, is the classical limit for a Bell violation.
When the state of the control qubit is too decohered, we can no longer violate a Bell inequality.
}
\label{img:decreasing_visiswitch}
\end{figure*}

In this section, we present two further tests of consistency of our experimental proof.

First, we decreased the entanglement of the joint control system by increasing the delay of the interferometer inside the source (see Appendix~\ref{Met-Sec:Source}). %In this case, each quantum-switch is still a causally non-separable process, but we are unable to violate a Bell inequality. 
%In this case, each individual quantum-switch could still produce a causally non-separable processes if the control subsystems were in a sufficiently coherent state.
The more mixed the state of the control system becomes, the smaller is the amount of violation of a Bell inequality with the target systems which we can achieve, up until reaching the threshold of non-violation. The Bell parameter versus the ``source visibility'' (i.e., the two-photon visibility in its anti-correlated basis) is plotted in Fig.~\ref{img:decreasing_visiswitch}\textbf{a}. The dashed line is a calculation of the expected Bell parameter, including the imperfect visibility of the two interferometers.  All the data points agree with the expected trend within error. The small step at an entanglement visibility of around $0.5$ was caused by a lower fringe visibility which increased the systematic error in setting the phases $\phi_1$ and $\phi_1 + \pi/4$ (see Fig.~\ref{img:HH}).

As a second test, we decreased the degree of causal non-separability of the two processes. To do this, we introduced distinguishing information between the paths corresponding to the orders ${U}_{i_\text{A}}\preceq {U}_{i_\text{B}}$ and ${U}_{i_\text{B}}\preceq {U}_{i_\text{A}}$ (in only one quantum-switch, squares in Fig.~\ref{img:decreasing_visiswitch}\textbf{b}, and in both simultaneously, circles in Fig.~\ref{img:decreasing_visiswitch}\textbf{b}) by lengthening one of the paths with respect to the other, effectively reducing the visibility of the interferometers comprising the quantum-switches. As this occurs, we transition from a superposition of temporal orders to a mixture of them (in other words, to a causally-separable process, which satisfies assumption III). If all three assumptions are met, one cannot violate a Bell inequality between the two systems. Indeed, we experimentally observe that as the visibility is decreased, the Bell parameter also decreases (Fig.~\ref{img:decreasing_visiswitch}, Panel \textbf{b}). %In this plot the dashed lines are linear fits to the experimental data.

\renewcommand{\baselinestretch}{1.2}
\bibliography{QDC_BIB}

\begin{thebibliography}{43}
\providecommand{\natexlab}[1]{#1}
\providecommand{\url}[1]{\texttt{#1}}
\expandafter\ifx\csname urlstyle\endcsname\relax
  \providecommand{\doi}[1]{doi: #1}\else
  \providecommand{\doi}{doi: \begingroup \urlstyle{rm}\Url}\fi

\bibitem[Ftn()]{Ftn1}
We test hypothesis IIb by placing the two operations individually in a single
  quantum switch, and we show that each operation is local. Under the
  assumption of a definite causal order, this test implies that the sequential
  action of the two operations is also local. We could not have directly tested
  the combined action of the two operations because such a test could only show
  that control and target are coupled, but would not provide any information
  about whether this is due to the violation of IIb, or due to the indefinite
  causal order, or both.

\bibitem[Abbott et~al.(2017)Abbott, Wechs, Costa, and
  Branciard]{Abbott2017genuinely}
Alastair~A. Abbott, Julian Wechs, Fabio Costa, and Cyril Branciard.
\newblock Genuinely multipartite noncausality.
\newblock \emph{{Quantum}}, 1:\penalty0 39, December 2017.
\newblock ISSN 2521-327X.
\newblock \doi{10.22331/q-2017-12-14-39}.
\newblock URL \url{https://doi.org/10.22331/q-2017-12-14-39}.

\bibitem[Allard~Gu{\'{e}}rin and Brukner(2018)]{Gu_rin_2018}
Philippe Allard~Gu{\'{e}}rin and {\v{C}}aslav Brukner.
\newblock Observer-dependent locality of quantum events.
\newblock \emph{New Journal of Physics}, 20\penalty0 (10):\penalty0 103031, oct
  2018.
\newblock \doi{10.1088/1367-2630/aae742}.

\bibitem[Ara\'ujo et~al.(2014)Ara\'ujo, Costa, and
  Brukner]{PhysRevLett.113.250402}
Mateus Ara\'ujo, Fabio Costa, and {\v C}aslav Brukner.
\newblock Computational advantage from quantum-controlled ordering of gates.
\newblock \emph{Phys. Rev. Lett.}, 113:\penalty0 250402, Dec 2014.
\newblock \doi{10.1103/PhysRevLett.113.250402}.

\bibitem[Ara\'ujo et~al.(2015)Ara\'ujo, Branciard, Costa, Feix, Giarmatzi, and
  Brukner]{NewJournPhys.17.102001}
Mateus Ara\'ujo, Cyril Branciard, Fabio Costa, Adrien Feix, Christina
  Giarmatzi, and {\v C}aslav Brukner.
\newblock Witnessing causal nonseparability.
\newblock \emph{New Journal of Physics}, 17\penalty0 (10):\penalty0 102001,
  2015.
\newblock \doi{10.1088/1367-2630/17/10/102001}.

\bibitem[Aspect et~al.(1982)Aspect, Dalibard, and Roger]{PhysRevLett.49.1804}
Alain Aspect, Jean Dalibard, and G\'erard Roger.
\newblock Experimental test of bell's inequalities using time-varying
  analyzers.
\newblock \emph{Phys. Rev. Lett.}, 49:\penalty0 1804--1807, Dec 1982.
\newblock \doi{10.1103/PhysRevLett.49.1804}.
\newblock URL \url{https://link.aps.org/doi/10.1103/PhysRevLett.49.1804}.

\bibitem[Barrett(2007)]{PhysRevA.75.032304}
Jonathan Barrett.
\newblock Information processing in generalized probabilistic theories.
\newblock \emph{Phys. Rev. A}, 75:\penalty0 032304, Mar 2007.
\newblock \doi{10.1103/PhysRevA.75.032304}.
\newblock URL \url{https://link.aps.org/doi/10.1103/PhysRevA.75.032304}.

\bibitem[Bell(1964)]{Bell1964}
J.~S. Bell.
\newblock On the einstein-podolsky-rosen paradox.
\newblock \emph{Physics}, 1\penalty0 (3):\penalty0 195--200, 1964.

\bibitem[Branciard(2016)]{ScienRep.6.26018}
Cyril Branciard.
\newblock Witnesses of causal nonseparability: an introduction and a few case
  studies.
\newblock \emph{Scientific Reports}, 6:\penalty0 26018, May 2016.
\newblock \doi{10.1038/srep26018}.

\bibitem[Branciard et~al.(2016)Branciard, Ara\'ujo, Feix, Costa, and
  Brukner]{1367-2630-18-1-013008}
Cyril Branciard, Mateus Ara\'ujo, Adrien Feix, Fabio Costa, and {\v C}aslav
  Brukner.
\newblock The simplest causal inequalities and their violation.
\newblock \emph{New Journal of Physics}, 18\penalty0 (1):\penalty0 013008,
  2016.
\newblock \doi{10.1088/1367-2630/18/1/013008}.
\newblock URL \url{http://stacks.iop.org/1367-2630/18/i=1/a=013008}.

\bibitem[Brukner(2014)]{NatPhys.10.259}
{\v C}aslav Brukner.
\newblock Quantum causality.
\newblock \emph{Nat. Phys.}, 10:\penalty0 259--263, Apr 2014.
\newblock \doi{10.1038/nphys2930}.

\bibitem[Brunner et~al.(2014)Brunner, Cavalcanti, Pironio, Scarani, and
  Wehner]{RevModPhys.86.419}
Nicolas Brunner, Daniel Cavalcanti, Stefano Pironio, Valerio Scarani, and
  Stephanie Wehner.
\newblock Bell nonlocality.
\newblock \emph{Rev. Mod. Phys.}, 86:\penalty0 419--478, Apr 2014.
\newblock \doi{10.1103/RevModPhys.86.419}.
\newblock URL \url{https://link.aps.org/doi/10.1103/RevModPhys.86.419}.

\bibitem[Chiribella(2012)]{PhysRevA.86.040301}
Giulio Chiribella.
\newblock Perfect discrimination of no-signalling channels via quantum
  superposition of causal structures.
\newblock \emph{Phys. Rev. A}, 86:\penalty0 040301, Oct 2012.
\newblock \doi{10.1103/PhysRevA.86.040301}.

\bibitem[Chiribella et~al.(2013)Chiribella, D'Ariano, Perinotti, and
  Valiron]{PhysRevA.88.022318}
Giulio Chiribella, Giacomo~Mauro D'Ariano, Paolo Perinotti, and Benoit Valiron.
\newblock Quantum computations without definite causal structure.
\newblock \emph{Phys. Rev. A}, 88:\penalty0 022318, Aug 2013.
\newblock \doi{10.1103/PhysRevA.88.022318}.

\bibitem[Clauser et~al.(1969)Clauser, Horne, Shimony, and
  Holt]{PhysRevLett.23.880}
John~F. Clauser, Michael~A. Horne, Abner Shimony, and Richard~A. Holt.
\newblock Proposed experiment to test local hidden-variable theories.
\newblock \emph{Phys. Rev. Lett.}, 23:\penalty0 880--884, Oct 1969.
\newblock \doi{10.1103/PhysRevLett.23.880}.
\newblock URL \url{https://link.aps.org/doi/10.1103/PhysRevLett.23.880}.

\bibitem[Dakic and Brukner(2010)]{GPT3}
Borivoje Dakic and {\v C}aslav Brukner.
\newblock \emph{Quantum Theory and Beyond: Is Entanglement Special?}
\newblock Contribution to ``Deep beauty'', Editor Hans Halvorson (Cambridge
  University Press, 2010.
\newblock \doi{10.1017/CBO9780511976971.011}.

\bibitem[Duan et~al.(2000)Duan, Giedke, Cirac, and Zoller]{PhysRevLett.84.2722}
Lu-Ming Duan, G.~Giedke, J.~I. Cirac, and P.~Zoller.
\newblock Inseparability criterion for continuous variable systems.
\newblock \emph{Phys. Rev. Lett.}, 84:\penalty0 2722--2725, Mar 2000.
\newblock \doi{10.1103/PhysRevLett.84.2722}.
\newblock URL \url{https://link.aps.org/doi/10.1103/PhysRevLett.84.2722}.

\bibitem[Feix et~al.(2015)Feix, Ara\'ujo, and Brukner]{PhysRevA.92.052326}
Adrien Feix, Mateus Ara\'ujo, and {\v C}aslav Brukner.
\newblock Quantum superposition of the order of parties as a communication
  resource.
\newblock \emph{Phys. Rev. A}, 92:\penalty0 052326, Nov 2015.
\newblock \doi{10.1103/PhysRevA.92.052326}.

\bibitem[Freedman and Clauser(1972)]{PhysRevLett.28.938}
Stuart~J. Freedman and John~F. Clauser.
\newblock Experimental test of local hidden-variable theories.
\newblock \emph{Phys. Rev. Lett.}, 28:\penalty0 938--941, Apr 1972.
\newblock \doi{10.1103/PhysRevLett.28.938}.
\newblock URL \url{https://link.aps.org/doi/10.1103/PhysRevLett.28.938}.

\bibitem[Giustina et~al.(2015)Giustina, Versteegh, Wengerowsky, Handsteiner,
  Hochrainer, Phelan, Steinlechner, Kofler, Larsson, Abell\'an, Amaya, Pruneri,
  Mitchell, Beyer, Gerrits, Lita, Shalm, Nam, Scheidl, Ursin, Wittmann, and
  Zeilinger]{PhysRevLett.115.250401}
Marissa Giustina, Marijn A.~M. Versteegh, S\"oren Wengerowsky, Johannes
  Handsteiner, Armin Hochrainer, Kevin Phelan, Fabian Steinlechner, Johannes
  Kofler, Jan-\AA{}ke Larsson, Carlos Abell\'an, Waldimar Amaya, Valerio
  Pruneri, Morgan~W. Mitchell, J\"orn Beyer, Thomas Gerrits, Adriana~E. Lita,
  Lynden~K. Shalm, Sae~Woo Nam, Thomas Scheidl, Rupert Ursin, Bernhard
  Wittmann, and Anton Zeilinger.
\newblock Significant-loophole-free test of bell's theorem with entangled
  photons.
\newblock \emph{Phys. Rev. Lett.}, 115:\penalty0 250401, Dec 2015.
\newblock \doi{10.1103/PhysRevLett.115.250401}.
\newblock URL \url{https://link.aps.org/doi/10.1103/PhysRevLett.115.250401}.

\bibitem[Goswami et~al.(2018)Goswami, Giarmatzi, Kewming, Costa, Branciard,
  Romero, and White]{PhysRevLett.121.090503}
K.~Goswami, C.~Giarmatzi, M.~Kewming, F.~Costa, C.~Branciard, J.~Romero, and
  A.~G. White.
\newblock Indefinite causal order in a quantum switch.
\newblock \emph{Phys. Rev. Lett.}, 121:\penalty0 090503, Aug 2018.
\newblock \doi{10.1103/PhysRevLett.121.090503}.
\newblock URL \url{https://link.aps.org/doi/10.1103/PhysRevLett.121.090503}.

\bibitem[Gu\'erin et~al.(2016)Gu\'erin, Feix, Ara\'ujo, and
  Brukner]{PhysRevLett.117.100502}
Philippe~Allard Gu\'erin, Adrien Feix, Mateus Ara\'ujo, and {\v C}aslav
  Brukner.
\newblock Exponential communication complexity advantage from quantum
  superposition of the direction of communication.
\newblock \emph{Phys. Rev. Lett.}, 117:\penalty0 100502, Sep 2016.
\newblock \doi{10.1103/PhysRevLett.117.100502}.

\bibitem[Hardy(2001)]{GPT2}
Lucien Hardy.
\newblock Quantum theory from five reasonable axioms.
\newblock \emph{Preprint at arXiv:0101012 [quant-ph]}, 2001.

\bibitem[Hardy(2011)]{GPT1}
Lucien Hardy.
\newblock \emph{Foliable Operational Structures for General Probabilistic
  Theories}, page 409–442.
\newblock Cambridge University Press, 2011.
\newblock \doi{10.1017/CBO9780511976971.013}.

\bibitem[Hensen et~al.(2015)Hensen, Bernien, Dreau, Reiserer, Kalb, Blok,
  Ruitenberg, Vermeulen, Schouten, Abellan, Amaya, Pruneri, Mitchell, Markham,
  Twitchen, Elkouss, Wehner, Taminiau, and Hanson]{Nature.526.682}
B.~Hensen, H.~Bernien, A.~E. Dreau, A.~Reiserer, N.~Kalb, M.~S. Blok,
  J.~Ruitenberg, R.~F.~L. Vermeulen, R.~N. Schouten, C.~Abellan, W.~Amaya,
  V.~Pruneri, M.~W. Mitchell, M.~Markham, D.~J. Twitchen, D.~Elkouss,
  S.~Wehner, T.~H. Taminiau, and R.~Hanson.
\newblock Loophole-free bell inequality violation using electron spins
  separated by 1.3 kilometres.
\newblock \emph{Nature}, 526:\penalty0 682--686, Oct 2015.
\newblock \doi{10.1038/nature15759}.
\newblock URL \url{http://dx.doi.org/10.1038/nature15759}.

\bibitem[Howell et~al.(2004)Howell, Bennink, Bentley, and
  Boyd]{PhysRevLett.92.210403}
John~C. Howell, Ryan~S. Bennink, Sean~J. Bentley, and R.~W. Boyd.
\newblock Realization of the einstein-podolsky-rosen paradox using momentum-
  and position-entangled photons from spontaneous parametric down conversion.
\newblock \emph{Phys. Rev. Lett.}, 92:\penalty0 210403, May 2004.
\newblock \doi{10.1103/PhysRevLett.92.210403}.
\newblock URL \url{https://link.aps.org/doi/10.1103/PhysRevLett.92.210403}.

\bibitem[Kwiat et~al.(1993)Kwiat, Steinberg, and Chiao]{PhysRevA.47.R2472}
P.~G. Kwiat, A.~M. Steinberg, and R.~Y. Chiao.
\newblock High-visibility interference in a bell-inequality experiment for
  energy and time.
\newblock \emph{Phys. Rev. A}, 47:\penalty0 R2472--R2475, Apr 1993.
\newblock \doi{10.1103/PhysRevA.47.R2472}.
\newblock URL \url{https://link.aps.org/doi/10.1103/PhysRevA.47.R2472}.

\bibitem[Lamehi-Rachti and Mittig(1976)]{PhysRevD.14.2543}
M.~Lamehi-Rachti and W.~Mittig.
\newblock Quantum mechanics and hidden variables: A test of bell's inequality
  by the measurement of the spin correlation in low-energy proton-proton
  scattering.
\newblock \emph{Phys. Rev. D}, 14:\penalty0 2543--2555, Nov 1976.
\newblock \doi{10.1103/PhysRevD.14.2543}.
\newblock URL \url{https://link.aps.org/doi/10.1103/PhysRevD.14.2543}.

\bibitem[MacLean et~al.(2017)MacLean, Ried, Spekkens, and Resch]{MacLean}
Jean-Philippe~W. MacLean, Katja Ried, Robert~W. Spekkens, and Kevin~J. Resch.
\newblock Quantum-coherent mixtures of causal relations.
\newblock \emph{Nature Communications}, 8:\penalty0 15149, May 2017.
\newblock \doi{10.1038/ncomms15149}.

\bibitem[Masanes et~al.(2014)Masanes, M\"uller, P\'erez-Garc\'ia, and
  Augusiak]{Masanes2011}
Ll. Masanes, M.~P. M\"uller, D.~P\'erez-Garc\'ia, and R.~Augusiak.
\newblock Entanglement and the three-dimensionality of the bloch ball.
\newblock \emph{Journal of Mathematical Physics}, 55\penalty0 (12):\penalty0
  122203, 2014.
\newblock \doi{10.1063/1.4903510}.
\newblock URL \url{https://doi.org/10.1063/1.4903510}.

\bibitem[Masanes and M\"uller(2011)]{1367-2630-13-6-063001}
Llu\'is Masanes and Markus~P. M\"uller.
\newblock A derivation of quantum theory from physical requirements.
\newblock \emph{New Journal of Physics}, 13\penalty0 (6):\penalty0 063001,
  2011.
\newblock \doi{10.1088/1367-2630/13/6/063001}.
\newblock URL \url{http://stacks.iop.org/1367-2630/13/i=6/a=063001}.

\bibitem[Miklin et~al.(2017)Miklin, Abbott, Branciard, Chaves, and
  Budroni]{1367-2630-19-11-113041}
Nikolai Miklin, Alastair~A Abbott, Cyril Branciard, Rafael Chaves, and
  Costantino Budroni.
\newblock The entropic approach to causal correlations.
\newblock \emph{New Journal of Physics}, 19\penalty0 (11):\penalty0 113041,
  2017.
\newblock \doi{10.1088/1367-2630/aa8f9f}.
\newblock URL \url{http://stacks.iop.org/1367-2630/19/i=11/a=113041}.

\bibitem[Oreshkov(2019)]{Oreshkov1801.07594v1}
Ognyan Oreshkov.
\newblock Time-delocalized quantum subsystems and operations: on the existence
  of processes with indefinite causal structure in quantum mechanics.
\newblock \emph{{Quantum}}, 3:\penalty0 206, December 2019.
\newblock ISSN 2521-327X.
\newblock \doi{10.22331/q-2019-12-02-206}.
\newblock URL \url{https://doi.org/10.22331/q-2019-12-02-206}.

\bibitem[Oreshkov and Giarmatzi(2016)]{Oreshkov_2016NJP}
Ognyan Oreshkov and Christina Giarmatzi.
\newblock Causal and causally separable processes.
\newblock \emph{New Journal of Physics}, 18\penalty0 (9):\penalty0 093020,
  2016.
\newblock \doi{10.1088/1367-2630/18/9/093020}.

\bibitem[Oreshkov et~al.(2012)Oreshkov, Costa, and Brukner]{NatComm.3.012316}
Ognyan Oreshkov, Fabio Costa, and {\v C}aslav Brukner.
\newblock Quantum correlations with no causal order.
\newblock \emph{Nat. Commun.}, 3:\penalty0 1092, Oct 2012.
\newblock \doi{10.1038/ncomms2076}.

\bibitem[Procopio et~al.(2015)Procopio, Moqanaki, Ara\'ujo, Costa,
  Alonso~Calafell, Dowd, Hamel, Rozema, Brukner, and Walther]{NatCommun.6}
Lorenzo~M. Procopio, Amir Moqanaki, Mateus Ara\'ujo, Fabio Costa, Irati
  Alonso~Calafell, Emma~G. Dowd, Deny~R. Hamel, Lee~A. Rozema, {\v C}aslav
  Brukner, and Philip Walther.
\newblock Experimental superposition of orders of quantum gates.
\newblock \emph{Nat. Commun.}, 6\penalty0 (7913), Aug 2015.
\newblock \doi{10.1038/ncomms8913}.

\bibitem[Rarity and Tapster(1990)]{PhysRevLett.64.2495}
J.~G. Rarity and P.~R. Tapster.
\newblock Experimental violation of bell's inequality based on phase and
  momentum.
\newblock \emph{Phys. Rev. Lett.}, 64:\penalty0 2495--2498, May 1990.
\newblock \doi{10.1103/PhysRevLett.64.2495}.
\newblock URL \url{https://link.aps.org/doi/10.1103/PhysRevLett.64.2495}.

\bibitem[Renner and Brukner(2021)]{MartinRenner:2021}
Martin~J. Renner and \ifmmode \check{C}\else~\v{C}\fi{}aslav Brukner.
\newblock Reassessing the computational advantage of quantum-controlled
  ordering of gates.
\newblock \emph{Phys. Rev. Research}, 3:\penalty0 043012, Oct 2021.
\newblock \doi{10.1103/PhysRevResearch.3.043012}.
\newblock URL \url{https://link.aps.org/doi/10.1103/PhysRevResearch.3.043012}.

\bibitem[Rowe et~al.(2001)Rowe, Kielpinski, Meyer, Sackett, Itano, Monroe, and
  Wineland]{Nature.409.791}
M.~A. Rowe, D.~Kielpinski, V.~Meyer, C.~A. Sackett, W.~M. Itano, C.~Monroe, and
  D.~J. Wineland.
\newblock Experimental violation of a bell's inequality with efficient
  detection.
\newblock \emph{Nature}, 409:\penalty0 791, Feb 2001.
\newblock \doi{10.1038/35057215}.
\newblock URL \url{http://dx.doi.org/10.1038/35057215}.

\bibitem[Rubino et~al.(2017)Rubino, Rozema, Feix, Ara{\'u}jo, Zeuner, Procopio,
  Brukner, and Walther]{Rubinoe1602589}
Giulia Rubino, Lee~A. Rozema, Adrien Feix, Mateus Ara{\'u}jo, Jonas~M. Zeuner,
  Lorenzo~M. Procopio, {\v C}aslav Brukner, and Philip Walther.
\newblock Experimental verification of an indefinite causal order.
\newblock \emph{Science Advances}, 3\penalty0 (3), 2017.
\newblock \doi{10.1126/sciadv.1602589}.
\newblock URL \url{http://advances.sciencemag.org/content/3/3/e1602589}.

\bibitem[Shalm et~al.(2015)Shalm, Meyer-Scott, Christensen, Bierhorst, Wayne,
  Stevens, Gerrits, Glancy, Hamel, Allman, Coakley, Dyer, Hodge, Lita, Verma,
  Lambrocco, Tortorici, Migdall, Zhang, Kumor, Farr, Marsili, Shaw, Stern,
  Abell\'an, Amaya, Pruneri, Jennewein, Mitchell, Kwiat, Bienfang, Mirin,
  Knill, and Nam]{PhysRevLett.115.250402}
Lynden~K. Shalm, Evan Meyer-Scott, Bradley~G. Christensen, Peter Bierhorst,
  Michael~A. Wayne, Martin~J. Stevens, Thomas Gerrits, Scott Glancy, Deny~R.
  Hamel, Michael~S. Allman, Kevin~J. Coakley, Shellee~D. Dyer, Carson Hodge,
  Adriana~E. Lita, Varun~B. Verma, Camilla Lambrocco, Edward Tortorici, Alan~L.
  Migdall, Yanbao Zhang, Daniel~R. Kumor, William~H. Farr, Francesco Marsili,
  Matthew~D. Shaw, Jeffrey~A. Stern, Carlos Abell\'an, Waldimar Amaya, Valerio
  Pruneri, Thomas Jennewein, Morgan~W. Mitchell, Paul~G. Kwiat, Joshua~C.
  Bienfang, Richard~P. Mirin, Emanuel Knill, and Sae~Woo Nam.
\newblock Strong loophole-free test of local realism.
\newblock \emph{Phys. Rev. Lett.}, 115:\penalty0 250402, Dec 2015.
\newblock \doi{10.1103/PhysRevLett.115.250402}.
\newblock URL \url{https://link.aps.org/doi/10.1103/PhysRevLett.115.250402}.

\bibitem[Simon(2000)]{PhysRevLett.84.2726}
R.~Simon.
\newblock Peres-horodecki separability criterion for continuous variable
  systems.
\newblock \emph{Phys. Rev. Lett.}, 84:\penalty0 2726--2729, Mar 2000.
\newblock \doi{10.1103/PhysRevLett.84.2726}.
\newblock URL \url{https://link.aps.org/doi/10.1103/PhysRevLett.84.2726}.

\bibitem[Zych et~al.(2019)Zych, Costa, Pikovski, and Brukner]{mag}
Magdalena Zych, Fabio Costa, Igor Pikovski, and {\v C}aslav Brukner.
\newblock Bell's theorem for temporal order.
\newblock \emph{Nature Communications}, 10:\penalty0 3772, 08 2019.
\newblock \doi{10.1038/s41467-019-11579-x}.
\newblock URL \url{https://doi.org/10.1038/s41467-019-11579-x}.

\end{thebibliography}
\bibliographystyle{plainnat}

\onecolumngrid

%\newpage
\section*{Supplementary Information}

\subsection*{I. Quantum-switch and causal inequalities}

The quantum-switch~\cite{NatCommun.6, Rubinoe1602589} has been shown not to violate causal inequalities, making it impossible to use such a violation as a theory-independent proof that the causal order of the operations in the quantum-switch is indefinite. Here, we briefly re-examine such reasoning following Refs. ~\cite{NewJournPhys.17.102001, Oreshkov_2016NJP}.

We introduce the $x$, $y$ and $z$ indices to refer, respectively, to the measurements choices of Alice, Bob and Charlie. We call $a$, $b$ and $c$ their respective measurement results. It is always possible to re-write $p(a,b,c|x,y,z)$ as 
\begin{equation}
p^{\text{switch}}(a,b,c|x,y,z) = p(c|a,b,x,y,z) \, p(a,b|x,y,z).
\end{equation}
It should be noticed that, regardless of the causal order between operations in Alice's and Bob's laboratory, the operation in Charlie's laboratory  always occurs after them. In other words, his operation is in the future light cone of both Alice's and Bob's operations. Thus, $a$ and $b$ cannot depend on $z$, so
\begin{equation}
p(a,b|x,y,z) = p(a,b|x,y).
\end{equation}
As previously observed, after tracing out Charlie's laboratory in the quantum-switch, the process of Alice and Bob is causally separable. Thus, one can rewrite $p(a,b|x,y)$ in the form of a convex mixture, obtaining
\begin{equation}
p^{\text{switch}}(a,b,c|x,y,z) = p(c|a,b,x,y,z) \bigl[\zeta \cdot p^{A \preceq B}(a,b|x,y)
+ (1-\zeta) \cdot p^{B \preceq A}(a,b|x,y)\bigr],
\end{equation}
with $\zeta \geq 0$. We can combine the probabilities $p^{A \preceq B}(a,b|x,y)$ ($p^{B \preceq A}(a,b|x,y)$) and $p(c|a,b,x,y,z)$ as a product of the probability respecting the order $A \preceq B$ ($B \preceq A$) with the probability respecting the order $\lbrace A,B \rbrace \preceq C$
\begin{equation}
p^{\text{switch}}(a,b,c|x,y,z) = \zeta \cdot p^{A \preceq B\preceq C}(a,b,c|x,y,z) + (1-\zeta) \cdot p^{B \preceq A\preceq C}(a,b,c|x,y,z).
\end{equation}
Therefore, the quantum-switch is a process whose probabilities have a ``causal model'', i.e., it can always be understood as arising from events that are causally ordered, or from a convex mixture of causally ordered events. Thus, it satisfies all causal inequalities.

\subsection*{II. Hidden local definite causal order}

In general, while experimental tests can be used to prove that the conjunction of the assumptions underlying a given no-go theorem does not describe the phenomenology observed within quantum mechanics, they do not provide information on which of the assumptions is to be discarded. In this experiment, the application of the Bell's theorem to temporal order allowed us to test a conjunction of all our assumptions; yet, in order to verify which assumptions are valid, additional tests on a single quantum-switch were necessary. This notwithstanding, it is worth noting that testing only one single quantum-switch would not have provided an as stringent information.
In fact, as we showed above the experimental data taken from a single quantum-switch cannot violate causal inequalities, and thus can be understood as arising from an underlying causal model, in the spirit of simulation of quantum statistics by hidden variables. Such a model generates statistics compatible with operations performed on a system in a definite order, or in a convex mixture therefrom. %For instance, applied to our experiment, the model may suggest a description of a photon that is located only in one of the interferometer's paths throughout the single quantum-switch, and it could determine, in advance, the behaviour of the photon in the last measurement station after it passed Alice's and Bob's operations.
In terms of probabilities, the statistics in the quantum-switch $p^{\text{switch}}(a,b,c|x,y,z)= \int \, d\lambda \, \rho(\lambda) \,  p^{\text{causal}}(a,b,c|x,y,z,\lambda)$, where $p^{\text{causal}}(a,b,c|x,y,z,\lambda) = p^{A  \preccurlyeq B \preccurlyeq C} (a,b,c|x,y,z)$ or $p^{B  \preccurlyeq A \preccurlyeq C} (a,b,c|x,y,z)$ or a classical mixture therefrom, could therefore be mimicked by an underlying causal hidden variable model. The statistics obtained measuring the double quantum-switch with entangled temporal orders rules out a {\em local} causal hidden variable model that allows for this description, i.e., its statistics is incompatible with
\begin{align}
\label{eqn:p_2switches}
p^{\text{2-switches}}&(a_1,b_1,c_1,a_2,b_2,c_2|x_1,y_1,z_1,x_2,y_2,z_2) = \notag\\
&= \int d\lambda  \, \rho(\lambda)  \, p^{\text{causal}}(a_1,b_1,c_1|x_1,y_1,z_1,\lambda)  \, p^{\text{causal}}(a_2,b_2c_2|x_2,y_2,z_2,\lambda).
\end{align}
In other words, the causal model is called ``local'' if the statistics %can be understood as originating from (a convex mixture of) operations performed in each local laboratory according to a definite causal order.
%Physically, this means that the correlations between the events in the two laboratories are 
is compatible with the assumptions that: (a) the order of events in each of the laboratories is definite but may be correlated by a common cause in the past, and (b) measurement choices may have only local influences.
Our experimental data rule out the models in Eq.~\eqref{eqn:p_2switches} for the special case where Alice and Bob both apply a single operation with a single outcome (unitary).

\subsection*{III. On the physical implementability of the quantum-switch}

Skepticism has been expressed about whether a tabletop experiment can demonstrate indefinite causal structures. In Ref.~\cite{MacLean}, the authors claim that it is not possible to implement the quantum-switch without ``exotic physical scenarios''. In particular, they argue that one would need a closed time-like curve, and even then such an implementation would be inconsistent, being able to generate logical contradictions such as the grandfather paradox.
These criticisms are based on the assumption %fact  that Ref.~\cite{MacLean}
that   %represents 
causal structures must be represented via directed graphs. %, which is an unnecessary assumption.
In this representation, the quantum-switch becomes a directed graph with a cycle, which could indeed be %is indeed 
inconsistent and could generate logical contradictions. 

The tension between directed acyclic graphs (DAG) and causal structures in the quantum-switch is akin to the tension between classical ``hidden'' variable theories and quantum theory. For example, in order to describe an interferometric experiment in terms of classical variables, one is forced to say that the interfering system follows some exotic trajectories or in some non-local manner follows two classical trajectories ``at once.'' However, within quantum
theory one interprets interferometric tests as demonstrating that the very assumption that a system does follow a definite path is violated.

The formal sense in which the causal order of applying  operations in a quantum-switch is non-classical has been recently studied in Ref.~\cite{Oreshkov1801.07594v1}. The motivation of that work was to understand where and when the operations happen in the quantum-switch,
%(and whether the resulting transformation is  different to standard quantum theory with definite order of application of the operations)
which is precisely the question brought up in the context of a DAG representation. The author shows that the operations applied on systems in a quantum-switch act on subsystems that are not localised in time, i.e., on `time-delocalised' subsystems. It is  further shown that standard quantum theory, without exotic closed timelike-curves, is compatible with such time-delocalised operations and that they indeed realise genuine non-separable quantum processes. The work also concludes that experimental realisations of the quantum-switch, including specifically its entangled version described in this work, are genuine realisations of such time-delocalised processes. In other words, there is a well-defined fashion in which temporal relations between the application of operations in a quantum-switch cannot be represented with DAGs. In fact, the Bell theorem for temporal order~\cite{mag} and its version in this work can be interpreted as a limitation on achievable correlations when operations acting on a quantum system can be embedded in a causal structure compatible with an underlying DAG (or a probabilistic mixture thereof).

Therefore, %following this reasoning brings one to
the suitable conclusion to draw is that 
the causal structure in the quantum-switch cannot be represented by a DAG since
the latter can only represent what are called \emph{definite} causal structures~\cite{1367-2630-19-11-113041}. What our work demonstrates experimentally is %precisely
that the quantum-switch represents an \emph{indefinite} causal structure incompatible with any DAG, just like experimental violations of Bell's inequalities show that there exist correlations incompatible with local hidden variables.

The authors of Ref.~\cite{MacLean} further argue that, in a genuine quantum-switch, operations must be performed in the same spatio-temporal regions in each term of the superposition, so that only  their order is swapped.
As mentioned above, Ref.~\cite{Oreshkov1801.07594v1} showed that this is not necessary: in the quantum-switch a single operation can be ``time-delocalised'' over two (or more) spatio-temporal regions. %To understand the function of the quantum-switch, it is crucial to make a distinction between performing a {\em single} operation delocalised over two times, and performing {\em two} operations each in a single time.
In the originally proposed implementation of the quantum-switch~\cite{PhysRevA.88.022318}, as well as in ours, one could in principle register the time at which the signal passes through each box, which would decohere the superposition and make the interference between the causal orders vanish. However, since we do observe coherence, such information does not exists (i.e., it is not stored in any physical degree of freedom, as this would alter the results of the experiment).
The above requirement of ``the same spatio-temporal regions'' whose order is simply swapped is in principle realisable in a gravitational implementation of the quantum-switch for a certain choice of coordinates. There, a massive object is prepared in a spatial superposition, which results in the causal order between two events being opposite in the two  superposed terms. More precisely, a choice of coordinates can be found such that gate $\mathcal{A}$ is performed at a single time (the proper time of a local clock). However, an alternative choice of coordinates may as well be done such that the gate $\mathcal{A}$ is performed in a superposition of different times (according to the coordinate time or the time of a distant observer) in different superposed terms, before and after the gate $\mathcal{B}$. Therefore, even in this gravitational case, it is always possible to make a choice of coordinates where the operations appear to be performed at different times in the different superposed terms. Thus, it is in fact insubstantial to argue whether the operations are ``really performed at the same times, and their order is swapped'', or they are ``merely performed in superposition of different times'', as this depends on the choice of coordinates in which one wishes to describe the scenario (see~\cite{mag, Gu_rin_2018} for further discussions). Furthermore, note that, contrary to the arguments of~\cite{MacLean}, this proposal does not allow for the information to travel back in time, nor does it require closed time-like curves, and it does not give rise to any logical paradoxes.

%%In addition, we would like to remind the reader that in the article where the quantum-switch was first introduced (Ref.~\cite{PhysRevA.88.022318}), its proposed implementation was based on an electronic circuit where physical wires connect two boxes $\mathcal{A}$ and $\mathcal{B}$. A quantum variable would then control the position of these wires such that the signal would flow through the circuit in either the $\mathcal{A} \rightarrow \mathcal{B}$ or $\mathcal{B} \rightarrow \mathcal{A}$ orders. If one conjectures that the quantum-switch is only physically realizable in very exotic physical scenarios, they would then need to make the strong claim that the originally proposed implementation of the quantum-switch was not genuine.

%%Our implementation follows closely the original proposal. The main differences is that we use photons, instead of electrons, in order not to need wires to move them around, and use a degree of freedom of the photon as the control system, in order not to need any macroscopic switching mechanism, and thus make the experiment feasible with present technology.

In Ref.~\cite{MacLean}, another criticism follows from the observation that the quantum-switch can easily be ``simulated'' by using additional copies of the boxes $\mathcal{A}$ and $\mathcal{B}$, as was already noted in the original proposal of the quantum-switch~\cite{PhysRevA.88.022318}. In particular, one could use an unfolded Mach-Zehnder interferometer, with gates $\mathcal{A}_1$ and $\mathcal{B}_1$ in one arm and gates $\mathcal{B}_2$ and $\mathcal{A}_2$ on the other. The straightforward response is that we %know  we 
do not use an unfolded Mach-Zehnder interferometer, but rather a folded one, and therefore that we use a single copy of each box instead of two. The number of applications of a box can operationally be determined by a counter (i.e., a ``flag'') that is raised each time the operation is applied on the system.
The very fact that an unfolded Mach-Zehnder interferometer requires two copies of each box to ``simulate'' the statistics of a folded one is a signature that the latter exhibits an indefinite causal order. Moreover, we also note that, in the unfolded version of the interferometer, it would be necessary to actively make $\mathcal{A}_1$ precisely equivalent to $\mathcal{A}_2$, whereas in our case this clearly follows from the implementation itself, as the gates are physically the same. It should also be emphasized that the present Bell-type proof of an indefinite causal order is valid even if the local gates are used more than once, as clarified in the main text. 

Furthermore, following the reasoning in Ref.~\cite{MacLean}, one could say that it is in principle possible to make the gate $\mathcal{A}$ ($\mathcal{B}$) act differently when it comes before or after $\mathcal{B}$ ($\mathcal{A}$), as in each case the photon passes through $\mathcal{A}$ ($\mathcal{B}$) at different times. This is true, but also applies to the originally proposed implementation. To make it \emph{locally} impossible, one could use the above mentioned superposition of a massive object to control the order of operations in two space-time regions. %\red{However, a choice of coordinates is always possible, even in this gravitational case, where the operations appear to be performed at different times in the different superposed terms. Thus, it is in fact insubstantial to argue whether the operations are ``really'' performed at the same times and their order is swapped or they are ``merely'' performed in superposition of different times --- as this depends on the choice of coordinates in which one wishes to describe the scenario (see \cite{mag} for further discussions).} %Nevertheless, it would be pointless to go through the trouble of putting the Earth in a superposition in order to make it impossible to do something that we do not do, and do not want to do.
In such a scenario, indeed, Alice (Bob) in her (his) local laboratory cannot make the gate $\mathcal{A}$ ($\mathcal{B}$) act differently in case the operation $\mathcal{A}$ ($\mathcal{B}$) is performed before or after $\mathcal{B}$ ($\mathcal{A}$). Nevertheless, a distant observer for whom Alice's (Bob's) operation happens in a superposition of two coordinate-times could make such contingency occur with a cleverly designed set-up (e.g., by sending a signal which triggers Alice's operation to change once it is received, as depicted in Fig.~\ref{img:Gravitational_switch}). As a consequence, as much as in the case of a table-top experiment the operation $\mathcal{A}$ can be made to act differently depending on whether it happens before or after $\mathcal{B}$, in its gravitational counterpart this can be achieved by a distant observer who triggers some change for certain time-coordinates. In conclusion, the requirement that operations $\mathcal{A}$ and $\mathcal{B}$ must \textit{even in principle} be forbidden to change depending on the order has no absolute meaning (i.e., it cannot be realized in all reference frames). Moreover, if the operations differed or their time was revealed (or stored in any degree of freedom), the results of the experiment would differ.

Finally, because of the differences highlighted above, one may object that the physics that describes a photonic quantum-switch is not equivalent to the one which is behind the gravitational quantum-switch for all possible observers. This is indeed correct. In fact, in the first case, the physics is described by Maxwell equations on Minkowski space-time, whereas in the latter case it is non-classical space-time that determines the dynamics.
%From an epistemologic standpoint, one could then try to interrogate the concepts of `simulation' and `analogue experimentations' to see if our work is best described in this fashion. Generally speaking, a simulation is understood as a `source' system described by a Hamiltonian that is isomorphic to that of the phenomenon under consideration \cite{arXiv:1712.05809v1} (which is, for technical reasons, inaccessible to the experimenter). There, a dynamical evolution of the source system strictly represents an evolution in the simulated phenomenon. On the other hand, analogue experimentation relies on an isomorphism between the equations of motion of the two systems, and there the dynamics are not one-to-one identical \cite{arXiv:1712.05809v1}. Note that, in both cases, the underlying physical dynamics may be described by completely different equations.
%In our case instead, 
However, although a local as well as a global observer could tell the difference between the gravitational and the photonic quantum-switch, it is not the case for a quantum particle which travels along the two superposed paths in either versions of the quantum-switch. In fact, in both cases the particle experiences a genuine quantum superposition of causal orders. And this is precisely the purpose of this experimental work: we do not aim to draw conclusions concerning global/local observers, but on the system undergoing the quantum process. Therefore, neither of the schemes (i.e., the gravitational and the photonic quantum-switch) is a ``simulation'' of one another. They are rather two equivalent representations of the dynamics experienced by a quantum particle in presence of a quantum superposition of causal orders, i.e., two representations of a quantum-switch.

%Finally, one could say that it is in principle possible to make the gate A act differently when it comes before or after B, as in each case the photon passes through A at different times. This is true, but also applies to the originally proposed implementation. To make it in principle impossible, one would need a superposition of spacetimes such as the one proposed in \cite{mag}. But it would be rather pointless to go through the trouble of putting the Earth in a superposition in order to make it impossible to do something that we do not do, and do not want to do.
\begin{figure}[th]
\centering
\includegraphics[width=0.7\columnwidth]{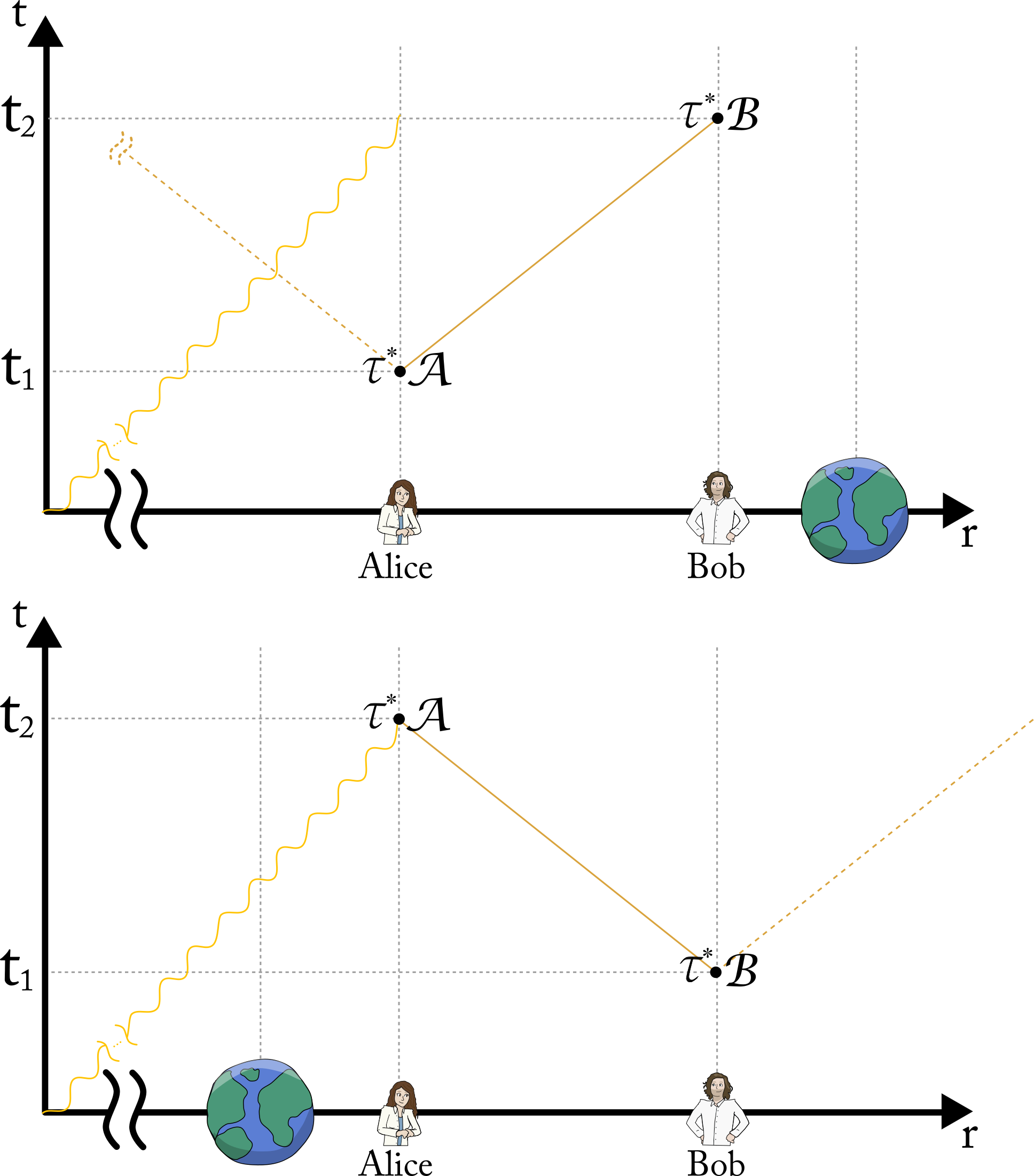}
\captionof{suppfigure}{\footnotesize \textbf{Schematic of a gravitational quantum-switch.} A quantum system is exchanged between Alice's and Bob's laboratories. The order in which such ``target'' system is exchanged is governed by a second system, a ``control'' system, which is encoded in the position of a massive object. By putting the massive object in a macroscopic superposition of two positions, one closer to Alice's and the other closer to Bob's position, one induces a relative time dilation between Alice's and Bob's laboratories. If an outside observer sends some system at a suitably chosen time, let us call it $t_{A \succ B}$, the observer could influence the functioning of the device that implements $\mathcal{A}$, e.g., when it acts second but not when it acts first, making the operation of Alice act different depending on the order.}
\label{img:Gravitational_switch}
\end{figure}

\begin{table}[ht]
\captionof{supptable}{\textbf{Comparison between the two-states probabilities $p(o_1,o_2|m_1,m_2, \omega^t_{1,2})$ and the products of marginal single-state probabilities $p(o_1|m_1, \omega^t_{1}) \cdot p(o_2|m_2, \omega^t_{2})$ for the input target states. - Part I.} The compatibility between the two sets of probabilities shows the separability of the input target state $\omega_{1,2}^t$. We indicate with H and V the states of horizontal and vertical polarization, with D and A the diagonal and anti-diagonal states, with R and L the circular polarization states right- and left-handed. The experimental error associated to each of these probabilities is $\pm 0.01$.}
\label{tab:sep1}
\centering
\begin{tabular}{c||c|c|c|c||c|c|c|c}
{\footnotesize Measur. Basis} & $p_{1,2}$ & $p_{1,2^{\perp}}$ & $p_{1^{\perp},2}$ & $p_{1^{\perp},2^{\perp}}$ & $p_1 \cdot p_2$ & $p_{1} \cdot p_{2^{\perp}}$ & $p_{1^{\perp}} \cdot p_2$ & $p_{1^{\perp}} \cdot p_{2^{\perp}}$\\
\cline{1-9}
H, H & 0.97 & 0.03 & 0.00 & 0.00 & 0.97 & 0.03 & 0.00 & 0.00 \\
H, V & 0.01 & 0.99 & 0.00 & 0.00 & 0.01 & 0.99 & 0.00 & 0.00 \\
H, A & 0.58 & 0.41 & 0.00 & 0.00 & 0.59 & 0.41 & 0.00 & 0.00 \\
H, D & 0.42 & 0.58 & 0.00 & 0.00 & 0.42 & 0.58 & 0.00 & 0.00 \\
H, R & 0.39 & 0.61 & 0.00 & 0.00 & 0.39 & 0.61 & 0.00 & 0.00 \\
H, L & 0.61 & 0.38 & 0.00 & 0.00 & 0.62 & 0.38 & 0.00 & 0.00 \\
V, H & 0.00 & 0.00 & 0.96 & 0.04 & 0.00 & 0.00 & 0.96 & 0.04 \\
V, V & 0.00 & 0.00 & 0.03 & 0.97 & 0.00 & 0.00 & 0.03 & 0.97 \\
V, A & 0.00 & 0.00 & 0.61 & 0.39 & 0.00 & 0.00 & 0.61 & 0.39 \\
V, D & 0.00 & 0.00 & 0.38 & 0.61 & 0.00 & 0.00 & 0.38 & 0.61 \\
V, R & 0.00 & 0.00 & 0.35 & 0.64 & 0.00 & 0.00 & 0.35 & 0.64 \\
V, L & 0.00 & 0.00 & 0.64 & 0.36 & 0.00 & 0.00 & 0.64 & 0.36 \\
A, H & 0.39 & 0.01 & 0.53 & 0.02 & 0.41 & 0.01 & 0.54 & 0.02 \\
A, V & 0.01 & 0.37 & 0.02 & 0.54 & 0.02 & 0.39 & 0.02 & 0.54 \\
A, A & 0.24 & 0.16 & 0.34 & 0.21 & 0.26 & 0.16 & 0.34 & 0.21 \\
A, D & 0.18 & 0.22 & 0.21 & 0.33 & 0.18 & 0.24 & 0.23 & 0.31 \\
A, R & 0.16 & 0.25 & 0.18 & 0.35 & 0.16 & 0.27 & 0.19 & 0.34 \\
A, L & 0.26 & 0.14 & 0.36 & 0.19 & 0.27 & 0.14 & 0.36 & 0.19 \\
D, H & 0.55 & 0.02 & 0.45 & 0.02 & 0.53 & 0.02 & 0.45 & 0.02 \\
D, V & 0.01 & 0.57 & 0.02 & 0.45 & 0.01 & 0.54 & 0.01 & 0.46 \\
D, A & 0.32 & 0.26 & 0.29 & 0.18 & 0.32 & 0.24 & 0.28 & 0.20 \\
D, D & 0.23 & 0.35 & 0.18 & 0.29 & 0.22 & 0.34 & 0.18 & 0.29 \\
D, R & 0.21 & 0.37 & 0.16 & 0.31 & 0.19 & 0.37 & 0.16 & 0.31 \\
D, L & 0.35 & 0.22 & 0.32 & 0.17 & 0.34 & 0.20 & 0.31 & 0.18 \\
R, H & 0.65 & 0.02 & 0.33 & 0.01 & 0.64 & 0.02 & 0.33 & 0.01 \\
R, V & 0.01 & 0.66 & 0.01 & 0.30 & 0.02 & 0.66 & 0.01 & 0.31 \\
R, A & 0.39 & 0.27 & 0.22 & 0.13 & 0.40 & 0.26 & 0.21 & 0.14 \\
R, D & 0.29 & 0.39 & 0.12 & 0.19 & 0.28 & 0.40 & 0.13 & 0.19 \\
R, R & 0.27 & 0.41 & 0.11 & 0.20 & 0.26 & 0.42 & 0.12 & 0.19 \\
R, L & 0.41 & 0.25 & 0.22 & 0.12 & 0.41 & 0.24 & 0.22 & 0.13 \\
L, H & 0.32 & 0.01 & 0.63 & 0.04 & 0.32 & 0.02 & 0.63 & 0.03 \\
L, V & 0.01 & 0.32 & 0.03 & 0.64 & 0.01 & 0.32 & 0.02 & 0.65 \\
L, A & 0.18 & 0.14 & 0.42 & 0.27 & 0.19 & 0.13 & 0.41 & 0.28 \\
L, D & 0.14 & 0.21 & 0.23 & 0.41 & 0.13 & 0.22 & 0.24 & 0.40 \\
L, R & 0.12 & 0.23 & 0.22 & 0.43 & 0.12 & 0.23 & 0.22 & 0.43 \\
L, L & 0.21 & 0.12 & 0.44 & 0.24 & 0.21 & 0.12 & 0.43 & 0.25
\end{tabular}
\end{table}

\begin{table}[ht]
\captionof{supptable}{\textbf{Comparison between the two-states probabilities $p(o_c,o_t|m_c,m_t, \omega_{1})$ and the products of marginal single-state probabilities $p(o_c|m_c, \omega^c_{1}) \cdot p(o_t|m_t, \omega^t_{1})$ for the control and the target states when only operation $U_{i_A}$ is acting on the input state.} We denoted as 0, 1, $+$, $-$, $l$ and $r$ the analogue of the polarization states H, V, D, A, L, R in the path degree of freedom. The two sets of probabilities associated to the control and the target states in output are compatible within experimental errors. The experimental error associated to each of these probabilities is $\pm 0.01$.}
\label{tab:Alice_TC}
\centering
\begin{tabular}{c|c||c|c|c|c||c|c|c|c}
{\footnotesize Meas. Basis} & {\footnotesize Prep.-Meas.} & $p_{c,t}$ & $p_{c,t^{\perp}}$ & $p_{c^{\perp},t}$ & $p_{c^{\perp},t^{\perp}}$ & $p_c \cdot p_t$ & $p_{c} \cdot p_{t^{\perp}}$ & $p_{c^{\perp}} \cdot p_t$ & $p_{c^{\perp}} \cdot p_{t^{\perp}}$\\
{\footnotesize (target)} & {\footnotesize Basis (control)} & & & & & & & & \\
\cline{1-10}
H & $+$ & 0.95 & 0.00 & 0.04 & 0.00 & 0.95 & 0.00 & 0.04 & 0.00 \\
D & $+$ & 0.47 & 0.48 & 0.01 & 0.03 & 0.47 & 0.49 & 0.02 & 0.02 \\
R & $+$ & 0.48 & 0.47 & 0.01 & 0.03 & 0.47 & 0.48 & 0.02 & 0.02 \\
H & $-$ & 0.07 & 0.00 & 0.92 & 0.01 & 0.07 & 0.00 & 0.91 & 0.01 \\
D & $-$ & 0.04 & 0.04 & 0.48 & 0.44 & 0.04 & 0.04 & 0.48 & 0.44 \\
R & $-$ & 0.04 & 0.04 & 0.41 & 0.51 & 0.04 & 0.04 & 0.41 & 0.51 \\
H & \textit{r} & 0.55 & 0.00 & 0.44 & 0.01 & 0.55 & 0.01 & 0.44 & 0.00 \\
D & \textit{r} & 0.20 & 0.26 & 0.28 & 0.26 & 0.22 & 0.24 & 0.26 & 0.28 \\
R & \textit{r} & 0.28 & 0.24 & 0.18 & 0.30 & 0.24 & 0.28 & 0.22 & 0.26 \\
H & \textit{l} & 0.50 & 0.00 & 0.50 & 0.00 & 0.50 & 0.00 & 0.50 & 0.00 \\
D & \textit{l} & 0.30 & 0.28 & 0.21 & 0.21 & 0.30 & 0.28 & 0.22 & 0.21 \\
R & \textit{l} & 0.27 & 0.30 & 0.20 & 0.23 & 0.27 & 0.30 & 0.20 & 0.23 \\
H & 0 & 0.51 & 0.00 & 0.49 & 0.00 & 0.50 & 0.00 & 0.49 & 0.00 \\
D & 0 & 0.26 & 0.30 & 0.26 & 0.17 & 0.30 & 0.27 & 0.23 & 0.21 \\
R & 0 & 0.28 & 0.26 & 0.23 & 0.23 & 0.27 & 0.26 & 0.24 & 0.23 \\
H & 1 & 0.56 & 0.00 & 0.43 & 0.01 & 0.56 & 0.00 & 0.44 & 0.00 \\
D & 1 & 0.27 & 0.29 & 0.22 & 0.23 & 0.27 & 0.29 & 0.21 & 0.23 \\
R & 1 & 0.28 & 0.28 & 0.17 & 0.27 & 0.25 & 0.31 & 0.20 & 0.24
\end{tabular}
\end{table}

\begin{table}[ht]
\captionof{supptable}{\textbf{Comparison between the two-states probabilities $p(o_c,o_t|m_c,m_t, \omega_{1})$ and the products of marginal single-state probabilities $p(o_c|m_c, \omega^c_{1}) \cdot p(o_t|m_t, \omega^t_{1})$ for the control and the target states when only operation $U_{i_B}$ is acting on the input state.} The two sets of probabilities associated to the control and the target states in output are compatible within experimental errors. The experimental error associated to each of these probabilities is $\pm 0.01$.}
\label{tab:Bob_TC}
\centering
\begin{tabular}{c|c||c|c|c|c||c|c|c|c}
{\footnotesize Meas. Basis} & {\footnotesize Prep.-Meas.} & $p_{c,t}$ & $p_{c,t^{\perp}}$ & $p_{c^{\perp},t}$ & $p_{c^{\perp},t^{\perp}}$ & $p_c \cdot p_t$ & $p_{c} \cdot p_{t^{\perp}}$ & $p_{c^{\perp}} \cdot p_t$ & $p_{c^{\perp}} \cdot p_{t^{\perp}}$\\
{\footnotesize (target)} & {\footnotesize Basis (control)} & & & & & & & & \\
\cline{1-10}
H & $+$ & 0.47 & 0.33 & 0.11 & 0.09 & 0.47 & 0.33 & 0.12 & 0.08 \\
D & $+$ & 0.50 & 0.27 & 0.18 & 0.06 & 0.51 & 0.25 & 0.16 & 0.08 \\
R & $+$ & 0.75 & 0.02 & 0.23 & 0.01 & 0.75 & 0.02 & 0.23 & 0.00 \\
H & $-$ & 0.11 & 0.15 & 0.49 & 0.25 & 0.16 & 0.11 & 0.44 & 0.29 \\
D & $-$ & 0.12 & 0.12 & 0.60 & 0.16 & 0.17 & 0.07 & 0.55 & 0.21 \\
R & $-$ & 0.25 & 0.01 & 0.67 & 0.07 & 0.24 & 0.02 & 0.68 & 0.06 \\
H & \textit{r} & 0.43 & 0.44 & 0.10 & 0.03 & 0.46 & 0.41 & 0.07 & 0.06 \\
D & \textit{r} & 0.54 & 0.32 & 0.09 & 0.05 & 0.54 & 0.32 & 0.09 & 0.05 \\
R & \textit{r} & 0.86 & 0.01 & 0.11 & 0.02 & 0.84 & 0.03 & 0.13 & 0.00 \\
H & \textit{l} & 0.16 & 0.06 & 0.49 & 0.29 & 0.14 & 0.08 & 0.51 & 0.27 \\
D & \textit{l} & 0.13 & 0.09 & 0.62 & 0.15 & 0.17 & 0.05 & 0.59 & 0.19 \\
R & \textit{l} & 0.20 & 0.01 & 0.73 & 0.05 & 0.20 & 0.01 & 0.73 & 0.05 \\
H & 0 & 0.26 & 0.28 & 0.26 & 0.19 & 0.29 & 0.26 & 0.24 & 0.22 \\
D & 0 & 0.40 & 0.14 & 0.41 & 0.05 & 0.44 & 0.10 & 0.37 & 0.09 \\
R & 0 & 0.48 & 0.04 & 0.42 & 0.06 & 0.47 & 0.05 & 0.44 & 0.05 \\
H & 1 & 0.32 & 0.23 & 0.29 & 0.15 & 0.34 & 0.22 & 0.27 & 0.17 \\
D & 1 & 0.32 & 0.24 & 0.32 & 0.12 & 0.36 & 0.20 & 0.28 & 0.16 \\
R & 1 & 0.56 & 0.00 & 0.41 & 0.03 & 0.54 & 0.02 & 0.43 & 0.01
\end{tabular}
\end{table}

% For your review copy (i.e., the file you initially send in for
% evaluation), you can use the {figure} environment and the
% \includegraphics command to stream your figures into the text, placing
% all figures at the end.  For the final, revised manuscript for
% acceptance and production, however, PostScript or other graphics
% should not be streamed into your compliled file.  Instead, set
% captions as simple paragraphs (with a \noindent tag), setting them
% off from the rest of the text with a \clearpage as shown  below, and
% submit figures as separate files according to the Art Department's
% instructions.

\end{document}